\pdfoutput=1  

\documentclass[11pt]{article}
\usepackage[margin=1in,left=1in]{geometry}

\usepackage{amsthm}
\usepackage{amssymb}
\usepackage{amsmath}
\usepackage{mathtools}
\usepackage{adjustbox}
\usepackage{xcolor}
\usepackage{stmaryrd}    
\usepackage[all]{xy}
\usepackage{rotating}
\usepackage{xfrac}
\usepackage{enumitem}
\setlist{nosep}

\usepackage{hhline}

\usepackage{tikz}
\usetikzlibrary{
  backgrounds, 
  decorations.pathmorphing, 
  decorations.markings,
  decorations.text
}
\usetikzlibrary{cd}

\usepackage{mathrsfs} 
\DeclareMathAlphabet{\mathpzc}{OT1}{pzc}{m}{it} 

\usepackage{hyperref}
\hypersetup{
  colorlinks = true,
  allcolors=.
}

\let\PLAINthebibliography\thebibliography
\renewcommand\thebibliography[1]{
  \PLAINthebibliography{#1}
  \setlength{\parskip}{0.5pt}
  \setlength{\itemsep}{0.5pt plus .3ex}
}

\usepackage[titletoc]{appendix}

\usepackage[safe]{tipa} 

\def\termscale{.9}
\newcommand\term[1]{\scalebox{\termscale}{$#1$}}

\newcommand{\ElectricFluxDensity}{E}

\definecolor{nyulight}{RGB}{107, 33, 158}

\definecolor{darkblue}{rgb}{0.05,0.25,0.65}
\definecolor{darkgreen}{RGB}{20,140,10}
\definecolor{lightgray}{rgb}{0.9,0.9,0.9}
\definecolor{darkorange}{RGB}{200,100,5}
\definecolor{darkyellow}{rgb}{.91,.91,0}

\usepackage{multirow}

\usepackage{enumerate} 

\newcommand{\yields}{\vdash}
\newcommand{\defneq}{\equiv}

\newcommand{\closed}{\mathrm{clsd}}

\newcommand{\Differential}{\mathrm{d}}
\newcommand{\differential}{\Differential}

\newcommand{\Sets}{
  \mathrm{Set}
}

\newcommand{\Groups}{
  \mathrm{Grp}
}

\newcommand{\RealNumbers}{
  \mathbb{R}
}

\newcommand{\ImaginaryUnit}{
  \mathrm{i}
}

\newcommand{\ComplexNumbers}{\mathbb{C}}

\newcommand{\CyclicGroup}[1]{\mathbb{Z}_{#1}}

\newcommand{\ZTwo}{
  \CyclicGroup{2}
}

\newcommand{\UnitaryGroup}{
  \mathrm{U}
}

\newcommand{\CircleGroup}{
  {\UnitaryGroup(1)}
}

\newcommand{\compact}{\cpt}

\newcommand{\proofstep}[1]{
  \mbox{\footnotesize #1}
}

\usepackage[new]{old-arrows}   

\newdir{> }{{}*!/10pt/@{>}}

\usepackage{amssymb}

\DeclareRobustCommand{\rchi}{{\mathpalette\irchi\relax}}
\newcommand{\irchi}[2]{\raisebox{\depth}{$#1\chi$}} 

\makeatletter
\newif\if@sup
\newtoks\@sups
\def\append@sup#1{\edef\act{\noexpand\@sups={\the\@sups #1}}\act}%
\def\reset@sup{\@supfalse\@sups={}}%
\def\mk@scripts#1#2{\if #2/ \if@sup ^{\the\@sups}\fi \else%
  \ifx #1_ \if@sup ^{\the\@sups}\reset@sup \fi {}_{#2}%
  \else \append@sup#2 \@suptrue \fi%
  \expandafter\mk@scripts\fi}
\def\tensor#1#2{\reset@sup#1\mk@scripts#2_/}
\def\multiscripts#1#2#3{\reset@sup{}\mk@scripts#1_/#2%
  \reset@sup\mk@scripts#3_/}
\makeatother

\makeatletter
\newbox\slashbox \setbox\slashbox=\hbox{$/$}
\def\itex@pslash#1{\setbox\@tempboxa=\hbox{$#1$}
  \@tempdima=0.5\wd\slashbox \advance\@tempdima 0.5\wd\@tempboxa
  \copy\slashbox \kern-\@tempdima \box\@tempboxa}
\def\slash{\protect\itex@pslash}
\makeatother

\def\clap#1{\hbox to 0pt{\hss#1\hss}}
\def\mathllap{\mathpalette\mathllapinternal}
\def\mathrlap{\mathpalette\mathrlapinternal}
\def\mathclap{\mathpalette\mathclapinternal}
\def\mathllapinternal#1#2{\llap{$\mathsurround=0pt#1{#2}$}}
\def\mathrlapinternal#1#2{\rlap{$\mathsurround=0pt#1{#2}$}}
\def\mathclapinternal#1#2{\clap{$\mathsurround=0pt#1{#2}$}}

\DeclareSymbolFont{symbolsC}{U}{txsyc}{m}{n}
\SetSymbolFont{symbolsC}{bold}{U}{txsyc}{bx}{n}
\DeclareFontSubstitution{U}{txsyc}{m}{n}

\newcommand{\TopologicalFields}{\mathrm{TopFields}_\Sigma}

\usepackage{cleveref}

\crefformat{section}{\S#2#1#3} 
\crefformat{subsection}{\S#2#1#3}
\crefformat{subsubsection}{\S#2#1#3}

\newtheorem{theorem}{Theorem}[section]
\newtheorem{lemma}[theorem]{Lemma}

\theoremstyle{definition}

\newtheorem{example}[theorem]{Example}

\newtheorem{remark}[theorem]{Remark}

\usepackage{amsfonts}
\usepackage{colortbl}

\renewcommand{\emph}{\textit}

\newcommand{\cpt}{{\cup\{\infty\}}}
\newcommand{\plus}{{\sqcup\{\infty\}}}

\begin{document}

\setlength{\abovedisplayskip}{2pt}
\setlength{\belowdisplayskip}{2pt}
\setlength{\abovedisplayshortskip}{-10pt}
\setlength{\belowdisplayshortskip}{2pt}

\title{Quantum observables of Quantized fluxes}

\author{
  Hisham Sati${}^{\ast \dagger}$
  \;\;
  and
  \;\;
  Urs Schreiber${}^{\ast}$
}

\maketitle

\thispagestyle{empty}

\begin{abstract}
While it has become widely appreciated that defining (higher) gauge theories requires, in addition to ordinary phase space data, also 
``flux quantization'' laws in generalized differential cohomology, there has been little discussion of the general rules, if any, for 
lifting Poisson-brackets of (flux-)observables and their quantization from traditional phase spaces to the resulting higher moduli 
stacks of flux-quantized gauge fields.

\smallskip 
In this short note, we present a systematic analysis of {\bf (i)} the canonical quantization of flux observables in Yang-Mills theory and {\bf (ii)} of valid flux 
quantization laws in abelian Yang-Mills, observing {\bf (iii)} that the resulting topological quantum observables form 
the homology Pontrjagin algebra of the loop space of the moduli space of flux-quantized gauge fields. 

\smallskip 
This is remarkable because the homology Ponrjagin algebra on loops of moduli makes immediate sense in broad generality for higher 
and non-abelian (non-linearly coupled) gauge fields, such as for the C-field in 11d supergravity, where it recovers the quantum 
effects previously discussed in the context of ``Hypothesis H''.
\end{abstract}

\vspace{.3cm}

\begin{center}
\begin{minipage}{11.5cm}
\tableofcontents
\end{minipage}
\end{center}

\vfill

\hrule
\vspace{5pt}

{
\footnotesize
\noindent
\def\arraystretch{1}
\tabcolsep=0pt
\begin{tabular}{ll}
${}^*$\,
&
Mathematics, Division of Science; and
\\
&
Center for Quantum and Topological Systems,
\\
&
NYUAD Research Institute,
\\
&
New York University Abu Dhabi, UAE.  
\end{tabular}
\hfill
\adjustbox{raise=-15pt}{
\includegraphics[width=3cm]{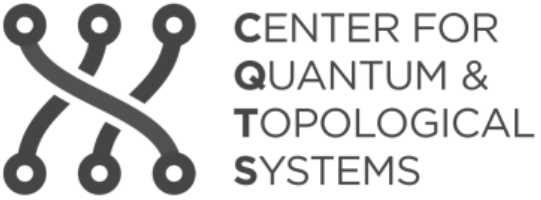}
}
}

\vspace{1mm} 
\noindent 
{
\footnotesize
${}^\dagger$The Courant Institute for Mathematical Sciences, NYU, NY
}

\vspace{.2cm}

\noindent
\scalebox{.82}{
The authors acknowledge the support by {\it Tamkeen} under the 
{\it NYU Abu Dhabi Research Institute grant} {\tt CG008}.
}

\medskip

\newpage

Electromagnetic fluxes in (higher) gauge theories are generally subject to ``quantization laws'' (``flux quantization'', ``charge quantization'', going back to \cite{Dirac31}\cite[\S 2]{Alvarez85} and others, we follow 
\cite[Intro]{Char}\cite{SS23PhaseSpace} with exposition in \cite{SS24FluxQuantization}, see also \cite{Freed00} 
and \cref{FluxQuantizationInAbelianYM} below), broadly in the sense of ``discretization'': In the familiar abelian case without ``self-sourcing'' of fluxes, these laws 
imply that total fluxes through (and hence charges inside) closed hypersurfaces form lattices (``charge lattices''), hence are integer multiples 
of certain unit fluxes (unit charges). This picture generalizes (with the recent construction of the non-abelian character map in \cite{Char}) 
to theories with non-abelian  fluxes (a famous example being the C-field in 11d supergravity, see \cite[Ex. 2.12]{SS24FluxQuantization}\cite{GSS24Supergravity}), now flux-quantized in non-abelian differential cohomology as discussed in the companion article \cite{SS23PhaseSpace}. In any case, such flux-quantization applies already to classical gauge fields (as soon as they serve as background fields for charged quantum probes).

\smallskip

On top of this, there is the actual quantization of fluxes, whereby quantum observables on fluxes form a non-commutative star-algebra (e.g. \cite[\S 6]{BogolyubovEtAl90}), 
reflecting quantum uncertainties (see \cref{QuantumObservablesOnYAngMillsFluxes}).

\smallskip

It should be clear that a deeper understanding of quantum gauge field theory requires an understanding of the combination of these two quantum effects, namely of quantum observables on quantized 
fluxes (cf. \cref{ConclusionAndOutlook}). However, existing discussions of the two aspects are mostly disjoint, among the exceptions being \cite{FMS07a}\cite{FMS07b}\cite{BBSS17} (to which we relate in the following, as we proceed).

\medskip

In this brief note we mean to clarify some general principles behind topological quantum observables on flux-quantized gauge fields, by making some observations  (Thms. \ref{PhaseSpaceOfYangMillsFluxes}, \ref{TopologicalQuantumObservablesOnYMFluxes} \& \ref{PontrjaginRingOfMaxwellFluxObservables}) 
which are not hard to prove but whose importance seems not to have been appreciated before, while they arguably touch on the heart of the matter.
Background discussion and proofs are relegated to \cref{BackgroundAndProofs}. 

\medskip

\noindent
{\bf Outline.}
\begin{itemize}[
  leftmargin=1.1cm,
  topsep=2pt,
  itemsep=2pt
]
\item[in \S\ref{QuantumObservablesOnYAngMillsFluxes}] we consider the non-perturbative (Rieffel $C^\ast$-algebraic) quantization of fluxes in ordinary Yang-Mills theory and observe that this involves making global choices and that the resulting star-algebra of {\it topological} observables is a convolution algebra of cohomology groups depending on these choices,

\item[in \S\ref{FluxQuantizationInAbelianYM}] we observe that for abelian Yang-Mills theory (electromagnetism) the same kind of topological choices parameterize the available ``flux quantization'' laws (including but going beyond the familiar Dirac charge quantization), in which form the situation generalizes to all ``higher'' (Maxwell-type) gauge theories such as found in higher dimensional supergravity theories,

\item[in \S\ref{PontrjaginAlgebrasOfQuantumObservables}] we bring these two observations together by showing that the algebra of topological quantum observables of (abelian) Yang-Mills theory (as per \S \ref{QuantumObservablesOnYAngMillsFluxes}) may equivalently be computed as the ``Pontrjagin homology algebra'' of the corresponding moduli spaces of quantized fluxes (as per \S\ref{FluxQuantizationInAbelianYM}), in which form this quantization prescription immediately generalizes to topological flux observables in higher gauge theories, thereby justifying some previous proposals on this matter,

\item[in \S\ref{ConclusionAndOutlook}] we close by pointing out a couple of further curious aspects which support the suggestion that the passage to Pontrjagin homology algebras of moduli of flux-quantized fields is ``the'' mechanism of quantization of topological flux observables in higher gauge theories.

\end{itemize}

\medskip
\medskip

\noindent
{\bf Acknowledgement.}
We thank Alberto Cattaneo for useful discussion.

\section{Quantum observables on Yang-Mills fluxes}
\label{QuantumObservablesOnYAngMillsFluxes}

\noindent
{\bf The phase space of fluxes.}
The following fundamental statement about classical observables on fluxes in Yang-Mills theory is implied by standard
facts about the phase space structure (cf. \cite[\S 3]{FriedmanPapastamatiou83}\cite[\S 2]{BassettoLazzizzeraSoldati84}) 
but seems not to have been noticed before, in its entirety (key observations are due to \cite{CattaneoPerez17}).
Consider $\mathfrak{g}$  a metric Lie algebra with pairing $\langle-,-\rangle$, $\Sigma$ a closed orientable surface (not necessarily connected) embedded in spacetime, and consider maps
$\alpha \,\in\, C^\infty(\Sigma, \mathfrak{g})$ 
as observables on fluxes that send the electric/magnetic flux density to its integral over
$\Sigma$ against $\langle\alpha,-\rangle$.

\begin{theorem}[Phase space of Yang-Mills fluxes]
  \label{PhaseSpaceOfYangMillsFluxes}
  The phase space of electromagnetic fluxes in $\mathfrak{g}$-Yang-Mills theory, through a closed orientable 
  surface $\Sigma$, is the Lie-Poisson manifold {\rm(e.g. \cite[\S 3]{Weinstein83})} associated with the Fr{\'e}chet Lie algebra
  of smooth maps into the semidirect product, via the adjoint action, of $\mathfrak{g}$ (with Lie bracket rescaled 
  by $\hbar \in \mathbb{R}_{> 0}$) on its underlying abelian Lie algebra $\mathfrak{g}_0$:
  \smallskip 
  \begin{equation}
    \label{LieAlgebraOfLinearYMFluxes}
    \underbrace{
    C^\infty\big(
      \Sigma
      ,\,
      (
        \mathfrak{g}_\hbar
        \ltimes_{{}_{\mathrm{ad}}}
        \mathfrak{g}_0
      )
    \big)
    }_{
      \mathclap{
      \scalebox{.7}{
        \color{gray}
        \bf
        \begin{tabular}{c}
          Linear observables
          on fluxes
        \end{tabular}
      }
      }
    }
    \;\;
    \simeq
    \;\;
    \underbrace{
    C^\infty\big(
      \Sigma
      \,,
      \mathfrak{g}_\hbar
    \big)
    }_{
      \mathclap{
      \scalebox{.7}{
        \rm
        \color{gray}
        \def\arraystretch{.9}
        \begin{tabular}{c}
          electric
        \end{tabular}
      }
      }
    }
    \ltimes_{{}_{\mathrm{ad}}}
    \underbrace{
      C^\infty\big(
        \Sigma
        \,,
        \mathfrak{g}_0
      \big)
    }_{
      \mathclap{
      \scalebox{.7}{
        \rm
        \color{gray}
        \def\arraystretch{.9}
        \begin{tabular}{c}
          magnetic
        \end{tabular}
      }
      }    
    }
    .
  \end{equation}
\end{theorem}
\begin{proof} We discuss this in \cref{BackgroundOnPhaseSpaceOfYangMillsFluxes}.\end{proof}

\medskip

\noindent
{\bf Quantum flux observables as group algebras.}
Thm. \ref{PhaseSpaceOfYangMillsFluxes} is remarkable, because while the full non-perturbative quantization of Yang-Mills theory famously remains an open problem, the theorem solves the case of the flux sector, since
the non-perturbative (aka ``strict-'' or ``$C^\ast$-algebraic-''deformation or ``Rieffel-'')  
quantizations (\cite{Rieffel89}\cite{Rieffel94}, review in \cite[\S 2]{Landsman99}\cite[\S 4]{LandsmanRamazan01}\cite[\S 2]{Hawkins08}) 
of Lie-Poisson phase spaces 
are well-known \cite{Rieffel90}\cite[Ex. 11.1]{LandsmanRamazan01} \cite[Ex. 2]{Landsman99}. Indeed, 
upon choosing a Lie group integrating the given Lie algebra, the non-perturbative quantum observables on the Lie-Poisson 
space form its {\it group algebra} under the convolution product, formed with due attention to analytic issues. 
One may think of this (cf. \cite[Ex. 11.3]{LandsmanRamazan01} and \cite{BinzHoneggerRieckers07}) as a version of the
time-honored quantization step from canonical commutation relations in the form of Heisenberg Lie algebras to their 
exponentiated Weyl form (\cite[p. 571]{vonNeumann31}) of quantum observables.

\smallskip

Concretely in the case of Thm. \ref{PhaseSpaceOfYangMillsFluxes}, choose $G$ a Lie group 
(not necessarily connected) with Lie algebra $\mathfrak{g}$, and choose a linear representation of $G$ on the underlying 
vector space of $\mathfrak{g}$ which on the connected component $G_{\mathrm{e}}$ restricts to the adjoint action. 
Then for $\Lambda \subset \mathfrak{g}$ a lattice (not necessarily of full dimension, in fact possibly zero) which is 
preserved under this action, we obtain the corresponding semidirect product Lie group of $G$ with the (partial) torus 
$\mathfrak{g}_0/\Lambda$ and hence a Fr{\'e}chet Lie group of maps Lie-integrating \eqref{LieAlgebraOfLinearYMFluxes}:
\smallskip 
\begin{equation}
  \label{LieGroupOfLinearYMFluxes}
  \underbrace{
  C^\infty
  \big(
    \Sigma
    ,\,
    G 
      \,\ltimes\, 
    (\mathfrak{g}_0 / \Lambda)
  \big)
  }_{
    \mathclap{
    \scalebox{.7}{
      \color{gray}
      \bf
      \def\arraystretch{.9}
      \begin{tabular}{c}
        Exponentiated linear 
        \\
        observables on fluxes
      \end{tabular}
    }
    }
  }
  \;\;
  \simeq
  \;\;
  \underbrace{
  C^\infty
  \big(
    \Sigma
    ,\,
    G 
  \big)
  }_{
    \mathclap{
      \scalebox{.7}{
        \color{gray}
        \bf
        \def\arraystretch{.9}
        \begin{tabular}{c}
          electric
        \end{tabular}
      }
    }
  }
  \;
  \ltimes
  \;
  \underbrace{
  C^\infty
  \big(
    \Sigma
    ,\,
    (\mathfrak{g}_0 / \Lambda)
  \big)
  }_{
    \mathclap{
      \scalebox{.7}{
        \color{gray}
        \bf
        \def\arraystretch{.9}
        \begin{tabular}{c}
          magnetic
        \end{tabular}
      }
    }  
  }
  \,.
\end{equation}

A typical example of the choices involved for $\mathfrak{g} = \mathfrak{su}(2)$ is given by $G \,\defneq\, \mathrm{SU}(2)$ 
and $\Lambda = 0 \subset \mathfrak{su}(2)$. However, it is important to notice the freedom of choosing $G$ to be non-connected, 
which here is part of the usual freedom in choosing quantizations.
For instance, already for $\mathfrak{g} = \mathfrak{u}(1) \,\simeq\, \mathbb{R}$ we may choose $G$ to be the direct product group 
$\mathrm{U}(1) \times \ZTwo$, $\Lambda = \mathbb{Z} \subset \mathbb{R}$ and the action of $\mathrm{U}(1) \times \ZTwo$ on 
$\mathfrak{u}(1) \simeq \mathbb{R}$ to factor through the $\ZTwo$-action by multiplication with $-1$. With this choice, the 
coefficient group in \eqref{LieGroupOfLinearYMFluxes} is the non-abelian group $\mathrm{U}(1) \times \ZTwo \ltimes \mathrm{U}(1)$, 
reflecting a non-trivial commutator between electric and magnetic flux observables;
see Ex. \ref{MoreGeneralStrictDeformationQuantizationOfMaxwellFluxes} below.

\medskip

\noindent
{\bf Topological quantum observables on fluxes.}
Here we are not concerned with the analytical fine-print of the convolution algebra on \eqref{LieGroupOfLinearYMFluxes}; 
instead, we focus on just its subsector of {\it topological flux observables}, namely those that are locally constant as functions on the 
group manifold, and as such form the subalgebra which is the ordinary group algebra \eqref{DiscreteConvolutionProduct}
of the group of connected components of 
\eqref{LieGroupOfLinearYMFluxes}. 
Interestingly, this group of connected components of \eqref{LieGroupOfLinearYMFluxes} is, by the smooth
Oka principle \eqref{SmoothOkaPrinciple}, equivalently a (possibly non-abelian) cohomology group of $\Sigma$:
\smallskip 
\begin{equation}
  \label{DiscreteGroupOfTopologicalYMFluxes}
  \def\arraystretch{1}
  \begin{array}{rcll}
  \underbrace{
  \pi_0
  \,
  C^\infty
  \big(
    \Sigma
    ,\,
    G 
      \,\ltimes\, 
    (\mathfrak{g}_0 / \Lambda)
  \big)  
  }_{
    \mathclap{
    \scalebox{.7}{
      \color{gray}
      \bf
      \def\arraystretch{.9}
      \begin{tabular}{c}
        Topological sectors of
        \\
        exponentiated linear
        \\
        observables on fluxes
      \end{tabular}
    }
    }
  }
  &\simeq&
  \pi_0
  \,
  \mathrm{Map}\big(
    \Sigma
    ,\,
    G 
      \ltimes
    B \Lambda
  \big)
  &
  \proofstep{
    by 
    \eqref{SmoothOkaPrinciple}
  }
  \\[-30pt]
  &\simeq&
  \pi_0
  \,
  \mathrm{Map}
  \big(
    \Sigma
    ,\,
    G 
  \big)    
  \,\ltimes\,
    \pi_0
    \,
    \mathrm{Map}
    \big(
      \Sigma
      ,\,
      B \Lambda
    \big)    
  &
  \proofstep{
    by
    \eqref{SmoothOkaIntoTorus}
  }
  \\[+6pt]
  &\simeq&
  \underbrace{
  H^0\big(
    \Sigma
    ;\,
    G
  \big)}_{
    \mathclap{
      \scalebox{.7}{
        \color{gray}
        \bf
        \def\arraystretch{.9}
        \begin{tabular}{c}
          electric
        \end{tabular}
      }
    }
  }
  \ltimes
  \underbrace{
  H^1\big(
    \Sigma
    ;\,
    \Lambda
  \big)
  }_{
    \mathclap{
      \scalebox{.7}{
        \color{gray}
        \bf
        \def\arraystretch{.9}
        \begin{tabular}{c}
          magnetic
        \end{tabular}
      }
    }  
  }
  &
  \proofstep{
    by
    \eqref{ZeroCohomologyWithCoefficientsInATopologicalGroup}.
  }
  \end{array}
\end{equation}
\noindent
In the last line, we retain the topology on the cohomology coefficients $G$; 
see Rem. \ref{OrdinaryCohomologyWithTopologicalGroupCoefficients}.

\smallskip

In conclusion, combining Thm. \ref{PhaseSpaceOfYangMillsFluxes} with Rieffel-quantization of Lie-Poisson structures yields:
\begin{theorem}[Non-perturbative topological quantum observables on Yang-Mills fluxes]
\label{TopologicalQuantumObservablesOnYMFluxes}
  The convolution group algebra \eqref{DiscreteConvolutionProduct} on the (possibly non-abelian) cohomology group \eqref{DiscreteGroupOfTopologicalYMFluxes}
  \smallskip
  \begin{equation}
    \label{GroupAlgebraOnCohomologyGroupReflectingTopologicalFLuxes}
    \mathbb{C}\big[
      \pi_0
      \,
      \mathrm{Map}(
        \Sigma
        ;\,
        G \ltimes B \Lambda
      )
    \big]
    \;\;
    \simeq
    \;\;
    \mathbb{C}\big[
      H^0(\Sigma;G)
      \ltimes
      H^1(\Sigma; \Lambda)
    \big]    
  \end{equation}

\smallskip 
\noindent is a subalgebra of topological 
observables in a non-perturbative Rieffel-quantization of the phase space from Thm. \ref{PontrjaginRingOfMaxwellFluxObservables}
of fluxes in $\mathfrak{g}$-Yang-Mills.
\end{theorem}

\begin{example}[Non-perturbative topological quantum observables on Maxwell fluxes]
  \label{AStrictDeformationQuantizationOfMaxwellFluxes}  
  For $\mathfrak{u}(1)$-Yang-Mills theory (vacuum Maxwell theory) an evident choice in \eqref{LieGroupOfLinearYMFluxes} of Lie group 
  $G$ and lattice $\Lambda$  is $G \,\defneq\, \mathrm{U}(1)$ and 
  $\Lambda = \mathbb{Z} \hookrightarrow \mathbb{R} \simeq_{{}_{\mathbb{R}}} \mathfrak{u}(1)$.
  In this case $H^0(\Sigma;\, G) \,\simeq\, H^1(\Sigma; \mathbb{Z})$ (see \eqref{ZerothCircleCohomology}) and
  also $H^1(\Sigma;\, \Lambda) = H^1(\Sigma;\, \mathbb{Z})$, so that the algebra of topological flux quantum observables 
  from \eqref{GroupAlgebraOnCohomologyGroupReflectingTopologicalFLuxes} is this group algebra:
  \smallskip 
  \begin{equation}
    \label{VanillaMaxwellFluxQuantumObservanles}
    \mathbb{C}\big[
      \mathrm{Map}\big(
        \Sigma
        ,
        \underbrace{\CircleGroup}_{
          \scalebox{.7}{
            \color{gray}
            electric
          }
        }
        \times \!\!
        \underbrace{\CircleGroup}_{
          \scalebox{.7}{
            \color{gray}
            magnetic
          }
        }
     \hspace{-2mm} \big)
    \big]
    \;\;
    \simeq
    \;\;
    \mathbb{C}\big[
      \underbrace{
        H^1(\Sigma;\,\mathbb{Z})
      }_{
        \scalebox{.7}{
          \color{gray}
          electric
        }
      }
      \times
      \underbrace{
        H^1(\Sigma;\,\mathbb{Z})
      }_{
        \scalebox{.7}{
          \color{gray}
          magnetic
        }
      }
    \big]
    \,.
  \end{equation}
  Notice that, while this algebra is commutative, it is in general distinct from (and non-isomorphic to) the algebra of classical 
  observables with its pointwise (non-convoluting) product. The pointwise product sees the topological flux sectors as 
  ``superselection'' sectors, whose Hilbert space decomposes as a direct sum indexed by electric and magnetic flux, such 
  that all observables are block-diagonal with respect to this decomposition.

  In contrast, the quantum algebra \eqref{VanillaMaxwellFluxQuantumObservanles} has operator representations by tuples of 
  unitary operators, mutually commuting with each other but each acting by {\it shifting by a unit} through the lattice of 
  topological sectors. (This result is different from the proposal in \cite[p. 20]{FMS07b}, but not unlike in spirit.)

Concretely, consider the simple case that $\Sigma \,\simeq\, \mathbb{T}^2$ is a torus. Then
  $H^1\big(\mathbb{T}^2;\, \mathbb{Z}\big) \,\simeq\, \mathbb{Z}_a \times \mathbb{Z}_b$, such that, under the
  identification with $H^0\big(\mathbb{T}^2; S^1\big)$ (see \eqref{ZerothCircleCohomology}), an element
  $\vec n \defneq (n_a, n_b) \in \mathbb{Z}_a \times \mathbb{Z}_b$ is the homotopy class of a $\mathfrak{u}(1)$-valued smearing 
  function \eqref{SmearingFunction}
  on $\mathbb{T}^2$ exponentiated to a $\CircleGroup$-valued function
  $\mathbb{T}^2 \to \CircleGroup$ which winds $n_{a}$ and $n_b$ times around $\CircleGroup \,\simeq\, S^1$ as its arguments
  travel once around one or the other nontrivial cycle of $\mathbb{T}^2$, respectively. Denoting by 
  $$
    \begin{tikzcd}[
      row sep=-3pt,
      column sep=10pt
    ]
      H^1\big(
        \mathbb{T}^2;\,\mathbb{Z}
      \big)
      \times
      H^1\big(
        \mathbb{T}^2;\,\mathbb{Z}
      \big)
      \ar[
        rr,
        hook
      ]
      &&
      \mathbb{C}\big[
      H^1\big(
        \mathbb{T}^2;\,\mathbb{Z}
      \big)
      \times
      H^1\big(
        \mathbb{T}^2;\,\mathbb{Z}
      \big)
      \big]
      \\
      \term{   
       \big(
        \vec n^{\,\mathrm{el}},
        \vec n^{\,\mathrm{mag}}
      \big)
      }
      \ar[
        rr,
        phantom,
        "{ \longmapsto }"
      ]
      &&
      \term{
      \mathcal{O}
      \big(
        \vec n^{\,\mathrm{el}},
        \vec n^{\,\mathrm{mag}}
      \big)
      }
    \end{tikzcd}
  $$
  the observable corresponding to these classes of smearing functions,
  the quantum (operator) product on them is
  \[
    \mathcal{O}\big( 
      \vec n^{\, \mathrm{el}}, 
      \vec n^{\, \mathrm{mag}} 
    \big)
    \cdot
    \mathcal{O}\big( 
    \vec m^{\mathrm{el}}, 
    \vec m^{\mathrm{mag}}
    \big)
       \;=\;
    \mathcal{O}\big( 
      \vec n^{\, \mathrm{el}}
      +
      \vec m^{\mathrm{el}}
      ,\,
      \vec n^{\, \mathrm{mag}}
      +
      \vec m^{\mathrm{mag}}
    \big)
    \,,
  \]
  
  \smallskip 
\noindent as befits observables in Weyl form (\cite[p. 571]{vonNeumann31}).

  Notice that (the exponentials of) the total flux observables $e^{\frac{\ImaginaryUnit}{\hbar}\int_{\mathbb{T}^2} E}$
  and $e^{\frac{\ImaginaryUnit}{\hbar}\int_{\mathbb{T}^2} F_A}$, whose smearing functions are constant,
  are summands of $\mathcal{O}\big( (0,0), (0,0) \big)$. From this, all other topological flux observables are obtained by acting with 
  the generators $\mathcal{O}\big( (\pm 1,0), (0,0) \big)$ and $\mathcal{O}\big( (0,\pm 1), (0,0) \big)$ etc., which observe the
  sectors of the first Fourier modes of fluxes through $\mathbb{T}^2$.
\end{example}

\begin{example}[Quantum observables on ``large fluxes'']
  \label{MoreGeneralStrictDeformationQuantizationOfMaxwellFluxes}
  In generalizing Ex. \ref{AStrictDeformationQuantizationOfMaxwellFluxes}, 
  notice that in quantizing we do have the freedom 
  of choosing in \eqref{LieGroupOfLinearYMFluxes} a non-connected Lie group $G$ with Lie algebra $\mathfrak{u}(1)$ -- such as 
  $G \,\defneq\, \mathrm{U}(1) \times \ZTwo$, even while retaining $\mathfrak{u}(1)/\mathbb{Z} \,\simeq\, \mathrm{U}(1)$ as the 
  gauge group of the magnetic fluxes, as usual. 
  If we think of the flux observables in  \eqref{LieAlgebraOfLinearYMFluxes} as $G$-gauge transformations --  indeed this is how they
  appear more manifestly below in \eqref{TopologicalKKFieldsExpressedOnSigma} -- then the non-connected components of $G$ correspond
  to what are known as ``large gauge transformations'' (e.g. \cite[p. 31]{HenneauxTeitelboim92}). In this vein here we may speak of ``large fluxes''.
    Now since $\ZTwo \simeq \mathrm{Aut}(\mathbb{Z})$ we may consider a non-trivial action of  $G$ on $\mathfrak{g}_0/\Lambda$ in 
  \eqref{LieGroupOfLinearYMFluxes}.
  This leads, via \eqref{GroupAlgebraOnCohomologyGroupReflectingTopologicalFLuxes}, to a quantum algebra of topological flux observables:
  \smallskip 
  \begin{equation}
    \mathbb{C}\big[
      \pi_0\,
      \mathrm{Map}\big(
        \Sigma
        ,\,
        (
        \underbrace{
        \CircleGroup \times \ZTwo
        }_{
          \scalebox{.7}{
            \color{gray}
            electric
          }
        }
        \ltimes
        \underbrace{\CircleGroup}_{
          \mathclap{
            \scalebox{.7}{
              \color{gray}
              magnetic
            }
          }
        }
        )
      \big)
    \big]
    \;\;
    \simeq
    \;\;
    \mathbb{C}\big[
      \underbrace{
      H^1(\Sigma;\, \mathbb{Z})
      \times
      H^0(\Sigma;\, \ZTwo)
      }_{
        \scalebox{.7}{
          \color{gray}
          \hspace{.7cm}
          electric
        }
      }
      \ltimes
      \underbrace{
      H^1\big(
        \Sigma
        ;\,
        \mathbb{Z}
      \big)
      }_{
        \scalebox{.7}{
          \color{gray}
          magnetic
        }
      }
    \big]
    \,,
  \end{equation}
  which is non-commutative, due to a non-trivial commutator between magnetic and large electric fluxes. 
  More generally, we could as well choose a semidirect product group $G \,\defneq\, \mathrm{U}(1) \rtimes \ZTwo$ for the electric fluxes, 
  in which case already the electric topological flux observables among themselves have non-trivial commutators,
  as is generally the case for non-abelian non-topological flux observables \eqref{LieAlgebraOfLinearYMFluxes}. 
    Yet more generally we may replace $\ZTwo$ by any discrete group equipped with a pair of homomorphisms $K \to \ZTwo$, 
    to obtain the following quantum algebras of topological fluxes:
    \smallskip 
  \begin{equation}
    \label{GeneralQuantizationOfTopologicalEMFluxes}
    \mathbb{C}\Big[
      \pi_0\,
      \mathrm{Map}\big(
        \Sigma
        ,\,
        (
        \underbrace{
        \CircleGroup \rtimes K
        }_{
          \scalebox{.7}{
            \color{gray}
            electric
          }
        }
        \ltimes
        \underbrace{\CircleGroup}_{
          \mathclap{
            \scalebox{.7}{
              \color{gray}
              magnetic
            }
          }
        }
        )
      \big)
    \Big]
    \;\;
    \simeq
    \;\;
    \mathbb{C}\Big[
      \underbrace{
      H^1(\Sigma;\, \mathbb{Z})
      \rtimes
      H^0(\Sigma;\, K)
      }_{
        \scalebox{.7}{
          \color{gray}
          \hspace{.7cm}
          electric
        }
      }
      \ltimes
      \underbrace{
      H^1\big(
        \Sigma
        ;\,
        \mathbb{Z}
      \big)
      }_{
        \scalebox{.7}{
          \color{gray}
          magnetic
        }
      }
    \Big]
    \,.
  \end{equation}
\end{example}

\medskip

In these examples, the choice of non-connected integrations of the gauge Lie algebra is, while certainly mathematically admissible,
unusual in traditional discussions of gauge theory. We next highlight, in \cref{FluxQuantizationInAbelianYM}, that such global 
choices of gauge group structure are part of the general phenomenon of flux quantization and as such have a clear
relevance that deserves attention.

\medskip




\medskip

\section{Flux quantization in abelian Yang-Mills}
\label{FluxQuantizationInAbelianYM}

The electromagnetic flux density (the Faraday tensor) $F \,\in\, \Omega^2_{\mathrm{dR}}(X^4)$ on a spacetime manifold $X^4$ (cf. \cite[Ex. 3.1]{SS23PhaseSpace}) can be thought of
 as a map $F :\ X \to \mathbf{\Omega}^2_{\closed}$ to the closed-differential form classifier in the topos of smooth sets 
 (cf. \cite[\S 1.2.3.2]{dcct}\cite[\S 2.3]{GiotopoulosSati23}\cite[p. 4]{Schreiber24}).
A fundamental and now classical insight into quantum gauge theory is the observation that this 
needs to be accompanied by a map $\chi \,:\, X \to \mathbf{B}^2 \mathbb{Z}$ to the classifying space for integral 2-cohomology and by 
a homotopy $\widehat{A}$ in the $\infty$-topos of smooth $\infty$-groupoids (cf. \cite[Prop. 1.24]{Char}\footnote{When we refer 
to equation-, definition-, proposition-, page-numbers in \cite{Char} we refer to the version published by World Scientific --- see \href{https://ncatlab.org/schreiber/show/The+Character+Map\#PublishedVersion}{\tt ncatlab.org/schreiber/show/The+Character+Map\#PublishedVersion} --- which 
differs from the numbering in the arXiv version (otherwise the content is the same).}, exposition in \cite{FSS14Stacky}\cite{Schreiber24}). 
This identifies the images of the two in real cohomology (see \cite[Ex. 9.4]{Char}):
\vspace{-2mm} 
\begin{equation}
  \label{OrdinaryDiracFluxQuantization}
  \begin{tikzcd}[row sep=5pt, column sep=huge]
    & 
    \mathbf{\Omega}^2_{\closed}
    \ar[
      dr,
      "{
        \scalebox{.7}{
          \color{gray}
          \bf
          de Rham map
        }
      }"{sloped}
    ]
    \ar[
      dd,
      Rightarrow,
      shorten=3pt,
      shift right=25pt,
      "{ \widehat{A} }",
      "{ \sim }"{swap, sloped}
    ]
    \\
    X^4
    \ar[
      ur, 
      "{ F }", 
      "{\ }"{swap, name=s}
    ]
    \ar[dr, "{ \rchi }"{swap}, "{\ }"{name=t}]
    &&
    \mathbf{B}^2 \mathbb{R}\;.
    \\
    &
    \underbrace{\mathbf{B}^2 \mathbb{Z}}_{
      B \CircleGroup
    }
    \ar[
      ur,
      "{
        \scalebox{.7}{
          \color{gray}
          \bf
          extension
        }
      }"{sloped},
      "{
        \scalebox{.7}{
          \color{gray}
          \bf
          of scalars
        }
      }"{sloped, swap}
    ]
  \end{tikzcd}
\end{equation}
This extra data exhibits {\it flux quantization} (often: ``Dirac charge quantization'', for further discussion and pointers 
see \cite[\S 2]{Alvarez85}\cite[\S 2]{Freed00}\cite[p. 4]{Char}\cite[\S 3.1]{SS23PhaseSpace}\cite{SS24FluxQuantization}) in that it ensures that the integrated magnetic flux through any 2-sphere submanifold 
$S^2 \hookrightarrow X$ is an integer
\medskip 
$
  \int_{S^2} F
  \;\in\;
  \mathbb{Z}
  \,,
$
counting the {\it number} of elementary magnetic solitons enclosed by $S^2$. 
For instance, if 
$$
  X 
  \;\defneq\;
  \mathbb{R}^{1,1}_+
  \wedge
  \mathbb{R}^2_{\cpt}
$$
is Minkowski spacetime with the ``point at infinity'' of a spatial hyperplane adjoined -- encoding the constraint that fields 
vanish at infinity along this plane, as is the case for a real laboratory magnetic field through a slab of material in the 
laboratory -- and if 
$$
  S^2 
    \;\simeq\; 
  \{0,0\} 
    \times 
  \mathbb{R}^2_\cpt \longhookrightarrow 
  X
$$
is the resulting sphere, then flux quantization reflects the experimentally observed phenomenon of integer numbers 
of Abrikosov vortices in a type-II superconducting material (cf. \cite[\S 2.1]{SS24FluxQuantization}).

\medskip

While this situation in Maxwell theory is commonly felt to be settled, it is not outright clear (and has hardly received consideration) which corresponding quantization 
condition is to be imposed on the {\it electric} flux density represented by the Hodge-dual 2-form $\star F$.
To even state the question properly, we need to get hold of $\star F$ as an independent flux variable that can be subjected to
flux quantization.

\medskip

\noindent
{\bf Premetric fluxes.}
This is accomplished by the equivalent ``pre-metric'' formulation of Maxwell's equations 
(\cite[\S 80]{Cartan24}, cf. \cite[Ex. 3.8]{Freed00}\cite{HehlItinObukhov16}\cite[Rem. 2.3]{BBSS17}\cite[Def. 1.16]{LazaroiuShahbazi22}\cite[Def. (3)]{LazaroiuShahbazi23}\cite[\S 3.1]{SS23PhaseSpace}\cite[\S 2.4]{SS24FluxQuantization}), which (in vacuum) subjects not one but two closed flux density variables
$(F,\, \ElectricFluxDensity) \,:\, X \to \mathbf{\Omega}^2_{\closed} \times \mathbf{\Omega}^2_{\closed}$ to a further {\it constitutive relation} 
enforcing their Hodge-duality on $X^4$:
\medskip 
\begin{equation}
  \label{HodgeDualityConditionOnEMFlux}
  \ElectricFluxDensity \,=\, \star_4 F
  \,.
\end{equation}
Trivially equivalent as this re-formulation is, it makes manifest the maximal ``decoupling'' of the {\it cohomological} from the {\it metric} content of Maxwell's equations, 
clearly suggesting that a choice of flux quantization of the pair $(F,\,\ElectricFluxDensity)$ needs to be made along the lines of \eqref{OrdinaryDiracFluxQuantization}, 
and only afterward the Hodge-duality condition \eqref{HodgeDualityConditionOnEMFlux} to be reimposed, in some fashion. 

Notice that this is just 
the same approach which in supergravity/string theory is known (\cite[\S 2.4]{SS24FluxQuantization}) as the ``duality-symmetric'' or ``democratic'' formulation (eg. \cite[p. 2]{MkrtchyanValach23}) 
underlying notably the common {\it Hypothesis K} that RR-fluxes are quantized in topological K-theory (see \cite[p.3]{GradySati22}\cite[\S 4.1]{SS24FluxQuantization} for pointers), and likewise for 
the fields in M-theory (see \cite[\S 4]{tcu}), here underlying the analogous ``Hypothesis H'' (\cite[\S 4.2]{SS24FluxQuantization}, see p. \pageref{HypothesisH} below). 

In fact, we may observe (we expand on this in the companion article \cite{SS23PhaseSpace}) that the pre-metric/democratic formulation of vacuum Maxwell
theory essentially coincides with its canonical phase space formulation (cf. \cref{BackgroundOnPhaseSpaceOfYangMillsFluxes}):

Namely on a globally hyperbolic spacetime $X^4 \,\simeq\, \mathbb{R}^{0,1} \times X^3$ in temporal gauge ($A_0 = 0$) the Faraday tensor in $\Omega^2(X^4)_{\closed}$ 
decomposes as a magnetic flux density $F \in \Omega^2(X^3)_\closed$ and a temporal component, whose Hodge dual is the electric 
flux density $E \,\in\, \Omega^2(X^3)_{\closed}$ which takes the role of the field's canonical momentum, and whose closure condition $\differential E = 0$ now plays the role of the {\it Gau{\ss} law}
constraint \eqref{GaussLaw}. This way, $F$ and $E$ are indeed independent field variables on $X^3$ and the constitutive relation \eqref{HodgeDualityConditionOnEMFlux}
is all absorbed into the prescription by which initial value data $(F,E)$ on $X^3$ induces temporal evolution in $X^4 = \mathbb{R}^{0,1} \times X^3$ (cf. \cite[\S 3.1]{SS23PhaseSpace}).

\medskip

\noindent
{\bf Phase space flux quantization.}
Hence we need to ask: What are the admissible flux quantization laws for flux densities  
$(F, \ElectricFluxDensity) : X^3 \to \mathbf{\Omega}^2_{\closed} \times \mathbf{\Omega}^2_{\closed}$?  The general answer is given in \cite{Char}: 
These are given by choices of topological spaces\footnote{
\label{OnRationalSpaces}
Here we restrict attention to classifying spaces $\mathcal{A}$  which are simply connected with finite-dimensional rational 
cohomology in each degree. This is not to get sidetracked by technical complications which, while of interest to the issue
of flux quantization, are more esoteric and beyond the intended scope of this note, cf. \cite[Rem. 5.1]{Char}.} $\mathcal{A}$ whose rationalization 
is equivalent to $B^2 \mathbb{Q} \times B^2 \mathbb{Q}$:
\begin{equation}
  \label{FluxQuantizationLawsForPremetricEM}
  \begin{tikzcd}
    \mathcal{A} 
    \ar[
      rr,
      "{
        \scalebox{.7}{
          \color{gray}
          \bf
          rationalization
        }
      }"{yshift=1pt}
    ]
    \ar[
      rrrr,
       rounded corners,
      to path={
           ([yshift=-00pt]\tikztostart.south)
        -- ([yshift=-14pt]\tikztostart.south)
        -- node[yshift=5pt]{
           \scalebox{.7}{
             \color{gray}
             \bf
             \scalebox{1}{$\RealNumbers$}-rationalization
           }
        }
           ([yshift=-12pt]\tikztotarget.south)
        -- ([yshift=-00pt]\tikztotarget.south)
      }
    ]
    &[+10pt]&[+10pt]
    B^2 \mathbb{Q}
    \times
    B^2 \mathbb{Q}
    \ar[
      rr,
      "{
        \scalebox{.7}{
          \color{gray}
          \bf
          extension
        }
      }",
      "{
        \scalebox{.7}{
          \color{gray}
          \bf
          of scalars
        }
      }"{swap}
    ]
    &&
    B^2 \mathbb{R}
    \times
    B^2 \mathbb{R}
    \mathrlap{\,.}
  \end{tikzcd}
\end{equation}

\smallskip 
\noindent For any such choice, the corresponding flux-quantized pre-metric gauge fields are given by homotopies of 
smooth $\infty$-groupoids of the following form (\cite[Def. 9.3]{Char}, but for our purpose here the reader need not 
further be concerned with the details of this construction):
\begin{equation}
  \label{GeneralPremetricDiracFluxQuantization}
  \begin{tikzcd}[
    row sep=7pt, column sep=huge
  ]
    & 
    \mathbf{\Omega}^2_{\closed}
    \times
    \mathbf{\Omega}^2_{\closed}
    \ar[
      dr,
      "{
        \scalebox{.6}{
          \color{gray}
          \bf
          de Rham map
        }
      }"{sloped}
    ]
    \ar[
      dd,
      shorten=10pt,
      shift right=40pt,
      Rightarrow,
      "{
        \; \widehat{A_{\mathrm{EM}}}      
      }",
      "{ \sim^{\phantom{i}} }"{swap, sloped}
    ]
    \\
    X^3
    \ar[
      ur, 
      "{ (F,\, \ElectricFluxDensity) }", 
      "{\ }"{swap, name=s}
    ]
    \ar[
      dr, 
      "{ 
        \rchi 
      }"{swap}, 
      "{\ }"{name=t}
    ]
    &&
    \mathbf{B}^2 \mathbb{R}
    \times
    \mathbf{B}^2 \mathbb{R}
    \\
    &
    \mathcal{A}
    \ar[
      ur,
      "{
        \scalebox{.7}{
          \color{gray}
          \bf
          \scalebox{1.2}{$\mathbb{R}$}-rationalization
        }
      }"{sloped, swap}
    ]
  \end{tikzcd}
\end{equation}

Notice that there are {\it many} available choices for such $\mathcal{A}$, and that each choice is a statement (a prediction)
about the corresponding physics. For instance, any connected topological space $Q$ all of whose homotopy groups are finite
(e.g. classifying spaces of finite groups)
\begin{equation}
  \label{RationalizationOfClassifyingSpaceOfFiniteGroup}
  K \,\in\,
  \mathrm{Grp}^{\mathrm{fin}}
  \hspace{.9cm}
  \yields
  \hspace{.9cm}
  \begin{tikzcd}[column sep=large]
    B K 
    \ar[
      rr,
      "{
        \scalebox{.7}{
          \color{gray}
          \bf
          rationalization
        }
      }"
    ]
    &&
    \ast
   \end{tikzcd}
\end{equation}
has trivial rationalization. This implies that with $A$ also any $A \times Q$ is an admissible flux quantization
law \eqref{FluxQuantizationLawsForPremetricEM} 
for pre-metric electromagnetism. 

\smallskip 
\noindent We make explicit some of the possible choices of electromagnetic flux quantization:

\begin{itemize}[leftmargin=.4cm]
\item
In view of \eqref{OrdinaryDiracFluxQuantization}, a suggestive choice may be the one that subjects $\star F$ to  same integral flux quantization as $F$
\begin{equation}
  \label{IntIntFluxQuantization}
  \mathcal{A} 
    \,\defneq\, 
  \underbrace{
    B \CircleGroup
  }_{
    \scalebox{.7}{
      \color{gray}
      magnetic
    }
  }
    \times 
  \underbrace{
    B \CircleGroup
  }_{
    \scalebox{.7}{
      \color{gray}
      electric
    }
  }
  \,,
\end{equation}
This corresponds to the choice made in \cite[(1.26)]{FMS07a}\cite[Rem. 2.3]{BBSS17}\cite[Def. 4.1]{LazaroiuShahbazi22}\cite[Def. 4.3]{LazaroiuShahbazi23} \footnote{
The articles \cite{LazaroiuShahbazi22}\cite{LazaroiuShahbazi23}, in the context of 4D supergravity, consider actually a yet more general form of flux quantization, where the the two factors in \eqref{IntIntFluxQuantization} may twist as one moves in the spacetime manifold.
}.
In \cite[(3.4)]{FMS07b} it says that it ``follows immediately'' from the Hodge self-duality of Maxwell's equations; but this is to presuppose
the answer to the question: Whether flux quantization laws retain all the symmetries of the underlying differential form data is a hypothesis that would ultimately need to be decided by experiment. Mathematically it is consistent, but so would be many other choices.

\item
Indeed, the mathematical physics literature commonly implies {\it no} further condition on the electric flux density $\star F$, apart from
it being the Hodge-dual of the magnetic flux density. In particular, common discussions of electromagnetism assume that the topological 
content of an EM-field configuration is all encoded in the class $\chi : X \to B^2 \mathbb{Z}$ of a single $\mathrm{U}(1)$-principal
bundle, nothing else.
This assumption is reflected in the choice
\begin{equation}
  \label{EMFluxQuantizationWithNoElectricCondition}
  \mathcal{A}
  \,\defneq\,
  \underbrace{
    B \CircleGroup
  }_{
    \scalebox{.7}{
      \color{gray}
      magnetic
    }
  }
  \times
  \underbrace{
    B^2 \mathbb{Q}
  }_{
    \scalebox{.7}{
      \color{gray}
      electric
    }
  }
  .
\end{equation}

\item
In either case, we highlight that the flux quantization laws \eqref{FluxQuantizationLawsForPremetricEM}
are subject to choices of pure torsion components in the classifying space. Even if we remain within the traditional assumption
that magnetic flux is classified exactly by $B \CircleGroup$ and that also electric flux should satisfy an integrality constraint as in \eqref{IntIntFluxQuantization}, there is still the freedom to postulate that electric flux is classified by a non-connected extension
of $\CircleGroup$, such as $\CircleGroup \rtimes K$ for any finite group $K$ with any action on $\CircleGroup$:
\begin{equation}
  \label{CrossedEMFluxQuantization}
  \mathcal{A}
  \;\simeq\;
  B\big(
    \underbrace{
      \CircleGroup
    }_{
      \mathclap{
        \scalebox{.7}{
          \color{gray}
          magnetic
        }
      }
    }
    \rtimes
  \underbrace{
    K
   \ltimes 
   \CircleGroup
  }_{
    \scalebox{.7}{
      \color{gray}
      electric
    }
  }
  \big)
  \,.
\end{equation}
(This freedom of choosing ``global'' non-abelian structure even in abelian Yang-Mills theory has also been observed, from a different angle, in \cite{LazaroiuShahbazi22}.)

We will see in \cref{PontrjaginAlgebrasOfQuantumObservables} that it is not entirely a coincidence that the group in \eqref{CrossedEMFluxQuantization} 
is the one controlling the topological flux observables in \eqref{GeneralQuantizationOfTopologicalEMFluxes}.

In realizing this coincidence, note that the duplication of fluxes to be independently flux-quantized in the ``duality symmetric'' formulation here in \S\ref{FluxQuantizationInAbelianYM} (namely magnetic and electric fluxes separately) matches the fact in \S\ref{QuantumObservablesOnYAngMillsFluxes} that {\it on phase space} the magnetic potential and electric field become independent of each other, one being the canonical momentum of the other. (This relation between duality-symmetric/pre-metric fluxes and passage to phase space is discussed in detail in \cite{SS23PhaseSpace}.)

\end{itemize}

\medskip

\noindent
{\bf Hypotheses about flux quantization laws.}
In listing the above examples, our aim is not to dwell on phenomenological questions of experimental quantum electromagnetism (though these are 
worthwhile, cf. \cite[p. 28]{FMS07b}\cite{KitaevMooreWalker07}), but rather to amplify the previously underappreciated mathematics 
parameterizing the space of consistent possibilities.
Namely, it is important to realize that analogous {\it choices} of flux quantization laws need to be made in higher gauge theories, 
and absent further rules any such choice is a {\it hypothesis} on the fundamental nature of these theories.

\medskip 
{\label{HypothesisK}}
Notably, concerning the popular statement that RR-flux forms in type I/II supergravity theory are quantized in topological K-theory 
(for pointers see \cite{GradySati22}\cite[\S 1.3]{SS23PhaseSpace}, we refer to this traditional hypothesis as ``Hypothesis K'' following \cite[Rem. 4.1]{SS23Defect}): There is secretly 
a choice that has been made and existing consistency checks of this choice rarely try to differentiate it from other possible choices.

\medskip 
\label{HypothesisH} Perhaps more importantly, a similar choice of flux-quantization law needs to be made when considering the $C$-field fluxes in $D=11$ 
supergravity. The similarly canonical-looking choice in this case is ``Hypothesis H'' (\cite[\S 2.5]{Sati13}\cite{FSS19TwistedCohomotopy}\cite{SS20Orientifold}\cite{FSS21HopfWZInHypothesisH}\cite{FSS21TwistedStringInHypothesisH}\cite{SS21M5Anomaly}\cite{FSS22TwistorialCohomotopy}\cite{SS23Mf}).
For this choice, we had observed \cite{SS22Configurations}\cite{CSS23WeightSystems} that the Pontrjagin algebra of the loop space 
of the moduli space of flux-quantized C-fields looks a lot like an algebra of non-perturbative topological quantum observables on (fluxes sourced by)
M-branes (see also \cref{ConclusionAndOutlook}).

\medskip 
Next, in \cref{PontrjaginAlgebrasOfQuantumObservables}, we observe that the analogous statement is already true for plain electromagnetism.

\medskip

Notice that another role of non-perturbative Rieffel quantization in relation to C-field fluxes has been considered in \cite{MathaiSati}.

\section{Pontrjagin algebras of Quantum observables}
\label{PontrjaginAlgebrasOfQuantumObservables}

In unification of \cref{QuantumObservablesOnYAngMillsFluxes} and \cref{FluxQuantizationInAbelianYM}, 
we observe here that algebras of topological flux observables \eqref{GroupAlgebraOnCohomologyGroupReflectingTopologicalFLuxes}
arise as the {\it homology  Pontrjagin algebras} of the loop spaces of moduli spaces of flux-quantized fields.

\medskip

\noindent
{\bf Topological fields and KK-Reduction.}
We consider  globally hyperbolic spacetimes which, besides the temporal direction, have a line factor singled out, i.e.,
are product spaces of the form
\begin{equation}
  \label{}
  \underbrace{
    \mathbb{R}^{0,1}
    \mathclap{\phantom{\big)}}
  }_{
    \scalebox{.7}{time}
  }
  \times
  \underbrace{
  \big(
    \mathbb{R}^1
    \times
    \Sigma
  \big)
  }_{
    \mathclap{
      \scalebox{.7}{
        Cauchy surface
      }
    }
  }
  \;\;=\;\;
  \mathbb{R}^{1,1}
  \times
  \Sigma
  \,.
\end{equation}
Here we may regard (by \eqref{TopologicalKKFieldsExpressedOnSigma} below) the $\mathbb{R}^1$-factor as a  ``decompactified'' KK-compactification
fiber (i.e., the fully non-perturbative situation, as KK-theory goes). Indeed, a ``spontaneous'' compactification is automatically implied if 
we consider topological fields vanishing at spatial infinity \eqref{ClasifyingMapsVanishingAtInfinity}, as usual \eqref{YMInstantonSectors} 
for solitonic fields, because:
\begin{equation}
  \label{TopologicalKKCompactification}
  \def\arraystretch{1.6}
  \begin{array}{ll}
    \mathrm{Map}^{\ast/}\Big(
      \mathbb{R}^{0,1}_{\plus}
      \wedge
      \big(
        \mathbb{R}^1
        \times
        \Sigma
      \big)_{\compact}
      ,\,
      B ( G \ltimes B \Lambda )
    \Big)
    &
    \proofstep{
      \def\arraystretch{.9}
      \def\tabcolsep{-6pt}
      \begin{tabular}{l}
      moduli space of topological fields
      \\
      vanishing at spatial infinity
      \end{tabular}
    }
    \\
    \underset{\mathrm{whe}}{
    \;\simeq\;
    }
    \mathrm{Map}^{\ast/}\Big(
      \big(
        \mathbb{R}^1
        \times
        \Sigma
      \big)_{\compact}
      ,\,
      B ( G \ltimes B \Lambda )
    \Big)
    &
    \proofstep{
      by
      \eqref{SmashProductWithContractibleManifold}
    }
    \\
    \;\simeq\;
    \mathrm{Map}^{\ast/}\Big(
      \big(
        S^1
        \wedge
       \Sigma_{\plus}
      ,\,
      B ( G \ltimes B \Lambda ) 
    \Big)
    &
    \proofstep{
      by
      \eqref{CompactificationOfProductOfCompactSpaceWithLine}.
    }
  \end{array}
\end{equation}
(Here we are assuming, just for brevity for exposition, that $\Sigma$ is already compact itself.)

\smallskip
Elementary as this is mathematically, it is somewhat remarkable as it exhibits the moduli space of topological fields as a loop space:
\begin{equation}
  \label{KKFieldsModuliSpaceAsLoopSpace}
  \def\arraystretch{1.6}
  \begin{array}{ll}
    \mathrm{Map}^{\ast/}\Big(
      \mathbb{R}^{0,1}_{\plus}
      \wedge
      \big(
        \mathbb{R}^1
        \times
        \Sigma
      \big)_{\compact}
      ,\,
      B (G \ltimes B \Lambda)
    \Big)
    &
    \proofstep{
      \def\arraystretch{.9}
      \def\tabcolsep{-6pt}
      \begin{tabular}{l}
      moduli space of topological fields
      \\
      vanishing at spatial infinity
      \end{tabular}
    }
    \\
    \underset{
      \mathrm{whe}
    }{
    \;\simeq\;
    }
    \mathrm{Map}^{\ast/}\Big(
      \big(
        S^1
        \wedge
       \Sigma_{\plus}
      ,\,
      B (G \ltimes B \Lambda)
    \Big)
    &
    \proofstep{
      by \eqref{TopologicalKKCompactification}
    }
    \\
    \;\simeq\;
    \mathrm{Map}^{\ast/}\Big(
      S^1
      ,\,
      \mathrm{Map}^{\ast/}
      \big(
        \Sigma_\plus
        ,\,
        B (G \ltimes B \Lambda)
      \big)
    \Big)
    &
    \proofstep{
      by \eqref{InternalHomIsoForPointedSpaces}
    }
    \\
    \;\simeq\;
    \Omega
    \,
    \mathrm{Map}\big(
      \Sigma
      ,\,
      B (G \ltimes B \Lambda)
    \big)
    &
    \proofstep{
      by
      \eqref{BasedLoopSpace}
      \& 
      \eqref{DisjointBasePointIsLeftAdjoint}.
    }
\end{array}
\end{equation}
But this also means that:
\begin{equation}
  \label{TopologicalKKFieldsExpressedOnSigma}
  \def\arraystretch{1.6}
  \begin{array}{ll}
    \mathrm{Map}^{\ast/}\Big(
      \mathbb{R}^{0,1}_{\plus}
      \wedge
      \big(
        \mathbb{R}^1
        \times
        \Sigma
      \big)_{\compact}
      ,\,
      B (G \ltimes B \Lambda)
    \Big)
    &
    \proofstep{
      \def\arraystretch{.9}
      \def\tabcolsep{-6pt}
      \begin{tabular}{l}
      moduli space of topological fields
      \\
      vanishing at spatial infinity
      \end{tabular}
    }
    \\
    \;\simeq\;
    \Omega
    \,
    \mathrm{Map}\big(
      \Sigma
      ,\,
      B ( G \ltimes B \Lambda)
    \big)
    &
    \proofstep{
      by
      \eqref{KKFieldsModuliSpaceAsLoopSpace}
    }
  \\
  \;\simeq\;
  \mathrm{Map}\big(
    \Sigma
    ,\, 
    (G \ltimes B \Lambda)
  \big)
  &
  \proofstep{
    by
    \eqref{LoopSpaceOfMappingIntoBG}.
  }
  \end{array}
\end{equation}
\vspace{-.3cm}

\noindent

\vspace{.2cm}
\hspace{-1cm}
\begin{tabular}{ll}
\begin{tabular}{l}
\begin{minipage}{5.5cm}
Hereby the topological fields on $\mathbb{R}^1 \times \Sigma$ are re-expressed as
fields on $\Sigma$, as befits a KK-reduction. (This is the based version of double 
dimensional reduction via {\it free} looping \cite[\S 2.2]{BMSS19}).
\end{minipage}
\\
{}
\\
\begin{minipage}{8cm}
{\bf Fusion of topological KK-fields.}
Unwinding the definitions, one sees that the operation of loop concatenation in \eqref{KKFieldsModuliSpaceAsLoopSpace} corresponds
to the ``fusion'' of field solitons in the KK-direction.

For abelian fields, this fusion of solitons is reflected in the addition of their charges.
\end{minipage}
\end{tabular}
&
\hspace{-3.3cm}
\adjustbox{
  scale=.93,
  raise=-3.3cm
}{
\begin{tikzpicture}

\begin{scope}[]
\draw[
  draw opacity=0,
  fill=lightgray
]
  (-2.8,1.5)
  rectangle (+2.8,-1.5);

\draw[
  white,
  line width=2.8pt,
]
 (-2.65, 1.5) to (-2.65,-1.5);
\draw[
  white,
  line width=2.8pt,
]
 (-2.45, 1.5) to (-2.45,-1.5);
\draw[
  white,
  line width=2.8pt,
]
 (-2.25, 1.5) to (-2.25,-1.5);
\draw[
  white,
  line width=2.8pt,
]
 (+2.65, 1.5) to (+2.65,-1.5);
\draw[
  white,
  line width=2.8pt,
]
 (+2.45, 1.5) to (+2.45,-1.5);
\draw[
  white,
  line width=2.8pt,
]
 (+2.25, 1.5) to (+2.25,-1.5);
\draw[
  dashed,
  latex-latex,
  line width=.2
]
  (-2.2,1.5) to (-2.2,-1.5);
\node at (-2.08,0) {
  \scalebox{.6}{
    $\Sigma$
  }
};
\draw[
  dashed,
  latex-latex,
  line width=.2
]
 (-2.6, -1.5) to 
 (+2.6, -1.5);
\node
  at (-3, -1.5) 
  {
    \scalebox{.5}{$\infty$}
  };
\node
  at (+3, -1.5) 
  {
    \scalebox{.5}{$\infty$}
  };
\node
  at (0,-1.65){
    \scalebox{.5}{$\mathbb{R}^1$}
  };
\begin{scope}[
  shift={(-1,.6)}
]
\shadedraw[
  draw opacity=0,
  inner color=gray,
  outer color=lightgray
]
  (0,0) circle (.6);
\shadedraw[
  draw opacity=0,
  inner color=black,
  outer color=lightgray!210
]
  (0,0) circle (.45);
\end{scope}
\begin{scope}[
  shift={(.9,.1)}
]
\shadedraw[
  draw opacity=0,
  inner color=gray,
  outer color=lightgray
]
  (0,0) circle (.6);
\shadedraw[
  draw opacity=0,
  inner color=black,
  outer color=lightgray!210
]
  (0,0) circle (.45);
\end{scope}
\begin{scope}[
  shift={(-.2,-.7)}
]
\shadedraw[
  draw opacity=0,
  inner color=gray,
  outer color=lightgray
]
  (0,0) circle (.6);
\shadedraw[
  draw opacity=0,
  inner color=black,
  outer color=lightgray!210
]
  (0,0) circle (.45);
\end{scope}
\end{scope}


\begin{scope}[shift={(6.0,0)}]
\draw[
  draw opacity=0,
  fill=lightgray
]
  (-2.8,1.5)
  rectangle (+2.8,-1.5);

\draw[
  white,
  line width=2.8pt,
]
 (-2.65, 1.5) to (-2.65,-1.5);
\draw[
  white,
  line width=2.8pt,
]
 (-2.45, 1.5) to (-2.45,-1.5);
\draw[
  white,
  line width=2.8pt,
]
 (-2.25, 1.5) to (-2.25,-1.5);
\draw[
  white,
  line width=2.8pt,
]
 (+2.65, 1.5) to (+2.65,-1.5);
\draw[
  white,
  line width=2.8pt,
]
 (+2.45, 1.5) to (+2.45,-1.5);
\draw[
  white,
  line width=2.8pt,
]
 (+2.25, 1.5) to (+2.25,-1.5);
\draw[
  dashed,
  latex-latex,
  line width=.2
]
  (-2.2,1.5) to (-2.2,-1.5);
\node at (-2.08,0) {
  \scalebox{.6}{
    $\Sigma$
  }
};
\draw[
  dashed,
  latex-latex,
  line width=.2
]
 (-2.6, -1.5) to 
 (+2.6, -1.5);
\node
  at (-3, -1.5) 
  {
    \scalebox{.5}{$\infty$}
  };
\node
  at (+3, -1.5) 
  {
    \scalebox{.5}{$\infty$}
  };
\node
  at (0,-1.65){
    \scalebox{.5}{$\mathbb{R}^1$}
  };
\begin{scope}[
  shift={(-.7,-.2)}
]
\shadedraw[
  draw opacity=0,
  inner color=gray,
  outer color=lightgray
]
  (0,0) circle (.6);
\shadedraw[
  draw opacity=0,
  inner color=black,
  outer color=lightgray!210
]
  (0,0) circle (.45);
\end{scope}
\begin{scope}[
  shift={(1,+.2)}
]
\shadedraw[
  draw opacity=0,
  inner color=gray,
  outer color=lightgray
]
  (0,0) circle (.6);
\shadedraw[
  draw opacity=0,
  inner color=black,
  outer color=lightgray!210
]
  (0,0) circle (.45);
\end{scope}
\end{scope}


\begin{scope}[
  shift={(3.2,-3.6)}
]
\draw[
  draw opacity=0,
  fill=lightgray
]
  (-2.8,1.5)
  rectangle (+2.8,-1.5);

\draw[
  white,
  line width=2.8pt,
]
 (-2.65, 1.5) to (-2.65,-1.5);
\draw[
  white,
  line width=2.8pt,
]
 (-2.45, 1.5) to (-2.45,-1.5);
\draw[
  white,
  line width=2.8pt,
]
 (-2.25, 1.5) to (-2.25,-1.5);
\draw[
  white,
  line width=2.8pt,
]
 (+2.65, 1.5) to (+2.65,-1.5);
\draw[
  white,
  line width=2.8pt,
]
 (+2.45, 1.5) to (+2.45,-1.5);
\draw[
  white,
  line width=2.8pt,
]
 (+2.25, 1.5) to (+2.25,-1.5);
\draw[
  dashed,
  latex-latex,
  line width=.2
]
  (-2.2,1.5) to (-2.2,-1.5);
\node at (-2.08,0) {
  \scalebox{.6}{
    $\Sigma$
  }
};
\draw[
  dashed,
  latex-latex,
  line width=.2
]
 (-2.6, -1.5) to 
 (+2.6, -1.5);
\node
  at (-3, -1.5) 
  {
    \scalebox{.5}{$\infty$}
  };
\node
  at (+3, -1.5) 
  {
    \scalebox{.5}{$\infty$}
  };
\node
  at (0,-1.5){
    \scalebox{.5}{\colorbox{white}{$\mathbb{R}^1$}}
  };
\begin{scope}[
  xscale=.5,
  shift={(-1.9,0)}
]
\begin{scope}[
  shift={(-1,.6)}
]
\shadedraw[
  draw opacity=0,
  inner color=gray,
  outer color=lightgray
]
  (0,0) circle (.6);
\shadedraw[
  draw opacity=0,
  inner color=black,
  outer color=lightgray!210
]
  (0,0) circle (.45);
\end{scope}
\begin{scope}[
  shift={(.9,.1)}
]
\shadedraw[
  draw opacity=0,
  inner color=gray,
  outer color=lightgray
]
  (0,0) circle (.6);
\shadedraw[
  draw opacity=0,
  inner color=black,
  outer color=lightgray!210
]
  (0,0) circle (.45);
\end{scope}
\begin{scope}[
  shift={(-.2,-.7)}
]
\shadedraw[
  draw opacity=0,
  inner color=gray,
  outer color=lightgray
]
  (0,0) circle (.6);
\shadedraw[
  draw opacity=0,
  inner color=black,
  outer color=lightgray!210
]
  (0,0) circle (.45);
\end{scope}
\end{scope}

\begin{scope}[
  xscale={.5},
  shift={(1.9,0)}
]
\begin{scope}[
  shift={(-.7,-.2)}
]
\shadedraw[
  draw opacity=0,
  inner color=gray,
  outer color=lightgray
]
  (0,0) circle (.6);
\shadedraw[
  draw opacity=0,
  inner color=black,
  outer color=lightgray!210
]
  (0,0) circle (.45);
\end{scope}
\begin{scope}[
  shift={(1,+.2)}
]
\shadedraw[
  draw opacity=0,
  inner color=gray,
  outer color=lightgray
]
  (0,0) circle (.6);
\shadedraw[
  draw opacity=0,
  inner color=black,
  outer color=lightgray!210
]
  (0,0) circle (.45);
\end{scope}
\end{scope}

\end{scope}

\draw[
  |->
]
  (3,-1.7) to 
  (3,-2);

\end{tikzpicture}
}
\end{tabular}

\noindent
This illustrates how the Pontrjagin product \eqref{HomologyPontrjaginAlgebraonLoopSpaceOfFields} is the pushforward in homology along concatenation of loops (of flux configurations): Thereby a pair of holomology classes which are (discrete) Dirac delta-distributions supported on loops corresponding to flux configurations as show in the top panel are sent to the Dirac delta distribution supported on the corresponding concatenated and then rescaled loop, as indicated in the bottom panel.

\medskip

\noindent
{\bf Pontrjagin algebra as Quantum observables.}
To see how to obtain quantum observables of such moduli, let us step back for a moment and reconsider the notion of quantum
observables in the most simplistic non-trivial case, namely on a set $S \,\in\, \Sets$ of disconnected (``superselection'') 
sectors of a physical system. The algebra of such quantum observables is just the linear span $\mathbb{C}[S]$ 
(whose canonical basis elements are the observations: ``system is in sector $s \in S$'') equipped with the $S$-element wise 
product of complex numbers, and regarded as a star-algebra under $S$-element wise complex conjugation.

\smallskip 
Now regard this simplistic case from a more sophisticated perspective by regarding $S \,\in\, \mathrm{Set} \hookrightarrow \mathrm{kTopSp}$ 
as a topological space that happens to carry a discrete topology. Then we may equivalently say that the observables span the {\it homology} 
$H_\bullet(S;\, \mathbb{C}) \,\simeq\, \mathbb{C}[S]$ of the space of sectors of configurations of the physical system. This is noteworthy, 
because the notion of homology makes sense, of course, for general topological spaces.
The main point is that, while homology groups do not generally form a natural algebra structure -- much less a star-algebra structure 
as required on quantum observables (e.g. \cite[\S 6]{BogolyubovEtAl90}) --  they do so on {\it loop spaces} \eqref{BasedLoopSpace} (generally on ``H-spaces''):

\smallskip 
To that end, let $\TopologicalFields$ be a topological moduli space of topological field configurations of a given physical system, equipped with 
some basepoint. Then the homology of its loop space becomes a star-algebra (in fact a Hopf algebra) whose product is induced by the 
concatenation of loops (``Pontrjagin product'' \cite{Pontrjagin39}\cite{BottSamelson53}, cf. \cite[pp. 287]{Hatcher02})
and whose star-involution (``dagger''-operation) is induced by reversal of loops:
\begin{equation}  \label{HomologyPontrjaginAlgebraonLoopSpaceOfFields}
  \def\arraystretch{1}
  \begin{array}{c}
 \begin{tikzcd}[
   row sep=4pt,
   column sep=40pt
 ]
     \Omega 
     \,
     \TopologicalFields
   \;\times\;
     \Omega 
     \,
     \TopologicalFields
   \ar[
     rr,
     "{
       \mathrm{conc}
     }",
     "{
       \scalebox{.7}{
         \color{gray}
         concatenate loops
       }
     }"{swap}
   ]
   &&
     \Omega 
     \,
     \TopologicalFields
   \\
   H_\bullet\big(
     \Omega 
     \,
     \TopologicalFields
     ;\,
     \mathbb{C}
   \big)
   \otimes
   H_\bullet\big(
     \Omega 
     \,
     \TopologicalFields
     ;\,
     \mathbb{C}
   \big)
   \ar[
     rr,
     "{ 
       (-) \cdot (-)
       \;\;
       :=
       \;\;
       \mathrm{conc}_\ast 
     }",
     "{
       \scalebox{.7}{
         \color{gray}
         Pontrjagin product
       }
     }"{swap}
   ]
   &&
   H_\bullet\big(
     \Omega 
     \,
     \TopologicalFields
     ;\,
     \mathbb{C}
   \big)  
   \mathrlap{\,,}
 \end{tikzcd}  
  \\
  \\  
  \begin{tikzcd}[
    row sep=4pt,
    column sep=55pt
  ]
    \Omega 
    \,
    \TopologicalFields
    \ar[
      rr,
      "{ \mathrm{rev} }",
      "{
        \scalebox{.7}{
          \color{gray}
          reverse loops
        }
      }"{swap}
    ]
    &&
    \Omega
    \,
    \TopologicalFields
    \\
    H_\bullet\big(
      \Omega
      \,
      \TopologicalFields
      ;\,
      \mathbb{C}
    \big)
    \ar[
      drr,
      end anchor={
        [xshift=-10pt, yshift=-10pt]
      },
      "{
        (-)^\dagger
      }"{sloped},
      "{
        \scalebox{.7}{
          \color{gray}
          Pontrjagin dagger
        }
      }"{swap, sloped}
    ]
    \ar[
      rr,
      "{
        \mathrm{rev}_\ast
      }",
      "{
        \scalebox{.7}{
          \color{gray}
          Pontrjagin antipode
        }
      }"{swap}
    ]
    &&
    H_\bullet\big(
      \Omega
      \,
      \TopologicalFields
      ;\,
      \mathbb{C}
    \big)
    \ar[
      d,
      "{
        \overline{(-)}
      }",
      "{
        \scalebox{.7}{
          \color{gray}
          \def\arraystretch{.9}
          \begin{tabular}{c}
            complex
            \\
            conjugation
          \end{tabular}
        }
      }"{swap}
    ]
    \\[25pt]
    &&
    H_\bullet\big(
      \Omega
      \,
      \TopologicalFields
      ;\,
      \mathbb{C}
    \big)
    \mathrlap{\,.}
  \end{tikzcd}
  \end{array}
\end{equation}

\medskip 

In these terms, we now obtain the following main observation of this note:

\begin{theorem}[Pontrjagin ring of topological quantum observables on Maxwell fluxes]
 \label{PontrjaginRingOfMaxwellFluxObservables}
 The Pontrjagin Hopf-algebra \eqref{HomologyPontrjaginAlgebraonLoopSpaceOfFields}
 of the moduli space of topological 
 $\mathfrak{u}(1)$-gauge fields \eqref{KKFieldsModuliSpaceAsLoopSpace} subject to the flux-quantization law
 \eqref{CrossedEMFluxQuantization} is, in degree=0, isomorphic to the Hopf algebra of topological quantum observables 
 from Ex. \ref{MoreGeneralStrictDeformationQuantizationOfMaxwellFluxes} on fluxes in  $\mathfrak{u}(1)$-Yang-Mills:
 $$
   H_0\bigg(
     \mathrm{Map}^{\ast/}
     \Big(
       \mathbb{R}^{0,1}_\plus
       \wedge
       \big(
         \mathbb{R}^1
         \times
         \Sigma
       \big)_{\compact}
       ,\,
       B\big(
         \CircleGroup^2
         \rtimes 
         K
       \big)
     \Big)
       ;\,
       \mathbb{C}
   \bigg)
   \;\simeq\;
   \mathbb{C}\big[
     H^1(\Sigma;\,\mathbb{Z})^2
     \rtimes
     H^0(\Sigma;\, K)
   \big]
   \,.
 $$
\end{theorem}
\begin{proof}
We discuss this in \cref{HomotopyTheoryOfTopologicalFieldSectors}.
\end{proof}

This observation seems noteworthy in that it obtains the topological quantum flux observables directly from the topology of the flux quantization law, 
short-cutting the analysis of Poisson brackets, and as such immediately generalizes to higher non-abelian gauge theories subject to rather more subtle
flux quantization laws, such as the RR-fluxes in string theory and the C-field fluxes in M-theory mentioned on p. \pageref{HypothesisK}. 
We comment on the potential impact in the following outlook \cref{ConclusionAndOutlook}.

\section{Conclusion and Outlook}
\label{ConclusionAndOutlook}

\noindent

Above we have focused on Pontrjagin algebras (of loop spaces of moduli spaces of topological gauge fields) in degree $\!\!=\!\!0$, 
showing that these reproduce non-perturbative topological quantum observables on fluxes. It is interesting to notice 
that also in higher degrees these Pontrjagin algebras look like (higher) algebras of (higher) quantum observables:

\smallskip

\noindent
{\bf Pontrjagin algebras as higher quantum algebras.}
The rational homotopy type of a simply connected topological space $\mathcal{A}$ (as in footnote \ref{OnRationalSpaces}) 
is all encoded in its rational Whitehead $L_\infty$-algebra $\mathfrak{l}{\mathcal{A}} \,\in\, L_\infty\mathrm{Alg}$ (dually its Sullivan model, see \cite[Prop. 5.11]{Char}), and the Pontrjagin algebra of its loop space 
is \cite[p. 262]{MilnorMoore65}\cite[Thm. 16.13]{FHT00} the universal graded enveloping algebra $U(-)$ 
of the underlying binary Whitehead Lie algebra bracket $[-,-]$ (a graded super-Lie algebra, but this generalizes to the enveloping $A_\infty$-algebra of the full Whitehead $L_\infty$-algebra \cite[Thm. 4.1]{Moreno-Fernandez22}):
\[
  \mbox{
    $\mathcal{A}$ 
    simply connected
  }
  \hspace{.7cm}
    \Rightarrow
  \hspace{.7cm}
  H_\bullet\big(
    \Omega
    \,
    \mathcal{A}
    ;\,
    \mathbb{R}
  \big)
  \;\simeq\;
  U\big(
    \mathfrak{l}\mathcal{A},
    [-,-]
  \big) 
  \,.
\]
Now, of course, the passage from Lie algebras to universal enveloping algebras is again the (now ``formal'') 
deformation quantization of the corresponding Lie-Poisson structure \cite[(.42)]{Gutt83}\cite[\S 22]{Gutt11}\cite[p. 3]{PenkavaVanhaecke00}, 
now all in a higher-geometric sense reflected in the grading.

\smallskip

For example, if $\mathcal{A} \,\defneq\, S^4$ is the homotopy type of the 4-sphere, then $\mathfrak{l}S^4$ has generators $v_3$ in
degree 3 and $v_6$ in degree-6, with non-trivial super-Lie bracket being the {\it M-theory gauge Lie algebra} $[v_3, v_3] = v_6$
(\cite[(2.5)]{CJLP98}\cite[\S 4]{tcu}\cite[Ex. 2.2]{SatiVoronov23}\cite[(27)]{SS24FluxQuantization}), whose graded universal enveloping algebra is hence
\vspace{-.3cm}
\begin{equation}
  H_\bullet\big(
    \Omega \,S^4
    ;\,
    \mathbb{C}
  \big)
  \;\simeq\;
  \mathbb{C}\big[v_3, v_6\big]\big/
  \big(
    v_3^2 - v_6
  \big)
  \,.
\end{equation}
This may be thought of as a quantum deformation of the cohomology of the 3-sphere $S^3 \hookrightarrow \Omega S^4$
just as, up to degree shifts, the {\it quantum cohomology} of $\mathbb{C}P^1$ (cf. \cite[p. 275]{Witten90}\cite{DGR10}):
\[
  QH^\bullet\big(
    \mathbb{C}P^1
    ;\,
    \mathbb{C}
  \big)
  \;\simeq\; 
  \mathbb{C}[v_2, v_4]\big/
  \big(
    v_2^2 - v_4
  \big)
\]
is a deformation of the ordinary cohomology of $\mathbb{C}P^1$ (and, of course, $S^3$ forms a circle bundle over $\mathbb{C}P^1$, the Hopf fibration, as befits an M-theory lift.)

\medskip

In view of Thm. \ref{PontrjaginRingOfMaxwellFluxObservables} this suggests that passage to Pontrjagin algebras of loop spaces of moduli spaces may generally be a valid form of quantization, 
at least for topological observables, applicable in particular also to higher gauge theories. Since this construction shortcuts
known forms of quantization, one is led to ask how to think of it as a quantization process in more detail: 


\medskip 
\noindent
{\bf Topological Light-cone quantization?} 
A key aspect of Pontrjagin algebras of loops in moduli spaces is \eqref{KKFieldsModuliSpaceAsLoopSpace} that their 
product operation corresponds to sequencing along a singled-out {\it spatial} direction in spacetime. But, since the operator product
order of quantum observables is well-known to reflect their {\it temporal} ordering \cite[p. 35]{Feynman42}\cite[p. 381]{Feynman48} 
(cf., e.g., \cite[pp. 33]{Nagaosa99}), it stands to reason that the Pontrjagin product on loop space homology regarded as a quantum operator-

\vspace{1.8pt}
\hspace{-.85cm}
\def\tabcolsep{3pt}
\begin{tabular}{ll}
\begin{minipage}{13.6cm}
product must be reflecting 
sequences of events which happen by {\it joint} progression along a ``time-axis'' {\it and} along an effectively periodic spatial
direction. This is of course the case in (``discretized'') light-cone quantization (review in \cite{Heinzl01}), such as famously 
used for non-perturbative quantization of Yang-Mills theory (review in \cite[\S 12]{FrishmanSonnenschein10})  and of sectors 
of M-theory (review in \cite[\S 3.9/7.9]{Ydri18}).
\end{minipage}
& \;\;\;
\adjustbox{raise=-.9cm}{
\begin{tikzpicture}[
    xscale=1.1,
    CoilColor/.store in=\coilcolor,CoilColor=black,
    Step/.store in=\Step,Step=0.1,
    Width/.store in=\Width,Width=0.4,
    Coil2/.style={
        decorate,
        decoration={
            markings,
            mark= between positions 0 and 1 step \Step 
            with {
                \begin{scope}[yscale=#1]
                    \pgfmathparse{int(\pgfdecoratedpathlength/28.45*100*\Step)}
                    \edef\Hight{\pgfmathresult}
                    \ifnum\pgfkeysvalueof{/pgf/decoration/mark info/sequence number}=1
                        \path (0,0)++(90: \Hight/200 and \Width) coordinate (b);
                    \fi
                    \ifnum\pgfkeysvalueof{/pgf/decoration/mark info/sequence number}>1
                        \coordinate (b) at (d);
                    \fi
                    \path (b) arc (90:-135: \Hight/200 and \Width) coordinate (a);
                    \path (b) arc (90:-45: \Hight/200 and \Width) coordinate (c);
                    \path (b)++(\Hight/100,0) coordinate (d);
                    \draw[fill,\coilcolor!70!black]
                        (c)
                            .. controls +(-0.175,0) and +(-0.275,0) .. (d)
                            .. controls +(-0.325,0) and +(-0.225,0) .. (c);
                    \draw[white,line width=2pt]
                        (b)
                            .. controls +(0.3,0) and +(0.2,0) .. (c);
                    \draw[fill,\coilcolor]
                        (b)
                            .. controls +(0.275,0) and +(0.175,0) .. (c)
                            .. controls +(0.225,0) and +(0.325,0) .. (b);
                \end{scope}
            }
        }
    }
]

 \draw[
   Coil2=2,
   CoilColor=black,
   Step=0.52,
  ] 
  (0,-.27) -- ++ (0,.75);

\draw[
  white,
  line width=1.2
]
  (-.1,-.4)
    ellipse
  (.804 and .2);
\draw
  (-.1,-.4)
    ellipse
  (.804 and .2);

 \draw[
   Coil2=2,
   CoilColor=black,
   Step=0.52,
  ] 
  (0,-1.11)  -- ++ (0,.75);

  \draw[-Latex]
    (-.9, -1.2)
    to node[yshift=5pt, xshift=4pt, sloped, pos=.7] {\scalebox{.6}{time}}
    (-.9,.8);

\draw[
  bend left=11,
  decorate,
  decoration={
    text along path,
    text align=center,
    text={|\tiny|
      periodic space
  }}
]
  (-.9,-.43) to (.7,-.445);
  
\end{tikzpicture}
}
\end{tabular}

\medskip 
Indeed, the discussion in \cite[\S 4.9]{SS22Configurations} with the result of \cite{CSS23WeightSystems} indicate that the 
method of Pontrjagin algebra quantization applied to topological fluxes sourced by intersecting M-branes reproduces the quantum states of transversal 
M2/M5-brane bound states in discretized light-cone quantization as previously discussed in the BMS matrix model. 

\smallskip

More recently, the result of \cite{GSS24-FluxOnM5}\cite{SS24-AbAnyons}\cite{SS24-TQBits} shows that this method of quantization of topological fluxes also brings out rigorously the widely expected (but previously somewhat informally justified) Chern-Simons link quantum observables on M5-branes, where the topological moduli of the cohomotopically flux-quantized self-dual tensor field on ``open'' M5-branes include the configuration space of points in a plane (inside the M5, transverse to given solitons) which the Pontrjagin-algebra turns into the group algebra of the braid group. 

\smallskip

These results suggest that this is not a coincidence, but part of a general role that Pontrjagin algebras play in quantization of fluxes in (higher) gauge theories, cf. \cite[\S 2]{SS23Intro}.

\appendix

\section{Background and Proofs}
\label{BackgroundAndProofs}

Here we spell out technical details and prove the claims in the previous sections.

\S\ref{BackgroundOnPhaseSpaceOfYangMillsFluxes} -- Phase space of fluxes in Yang-Mills theory

\S\ref{HomotopyTheoryOfTopologicalFieldSectors} -- Homotopy theory of topological field sectors

\subsection{Phase space of Yang-Mills fluxes}
\label{BackgroundOnPhaseSpaceOfYangMillsFluxes}

\noindent
{\bf Phase space of Yang-Mills theory.} We start by recalling the canonical phase space structure of Yang-Mills theory 
in temporal gauge (e.g. \cite[\S 3]{FriedmanPapastamatiou83}\cite[\S 2]{BassettoLazzizzeraSoldati84}), following the insightful 
account of \cite{CattaneoPerez17}.\footnote{The discussion in \cite{CattaneoPerez17} is motivated by Ashtekar's phase space of first-order 
Einstein-gravity, which famously coincides with that of $\mathrm{SU}(2)$-Yang-Mills theory subject to further constraints. 
But these further constraints play no role in \cite{CattaneoPerez17} and the restriction to $G \defneq \mathrm{SU}(2)$ is inessential otherwise.
} We slightly generalize these accounts by admitting any metric Lie algebra as gauge algebra, and by admitting non-trivial 
topological sectors of gauge potentials.

\medskip

\noindent
Consider:
\begin{itemize}[leftmargin=.5cm]
\setlength\itemsep{2pt}
\item
$\mathfrak{g}$ a finite-dimensional real metric Lie algebra, hence equipped with a symmetric $\mathrm{ad}$-invariant 
non-degenerate (but not necessarily definite) bilinear form
$
  \langle -,-\rangle
  \,:\,
  \mathfrak{g} \otimes \mathfrak{g}
  \to
  \mathbb{R}
  \,.
$

\item
$G$ a Lie group with Lie algebra $\mathfrak{g}$, being the Yang-Mills gauge/structure group.

\item 
$X$ a smooth 3-manifold, thought of as a Cauchy surface in globally hyperbolic spacetime $\mathbb{R}^{0,1} \times X$.

\item 
  $\widehat X \twoheadrightarrow X$ any differentiably good open cover,
  i.e., $\hat X \,\defneq\, \underset{i \in I}{\sqcup} U_i$ for $\big\{U_i \xhookrightarrow{ \iota_i } X \big\}_{i \in I}$ a set 
  of open subsets which cover, $\underset{i \in I}{\cup} U_i \,=\, X$, and all whose non-empty finite intersections are diffeomorphic
  to an open ball.

\item 
  $g_{\bullet,\bullet} \,:\, \widehat{X} \times_X \widehat X \to G$ a smooth (``transition''-)function, hence with components $g_{i j} \,:\, U_i \cap U_j \to G$,
  such that on any $U_i \cap U_j \cap U_k$ we have
  $g_{i j} \cdot g_{j k} = g_{i k}$,
  encoding a topological sector (a $G$-principal bundle) $P$ of $G$-Yang-Mills theory on $\mathbb{R}^{0,1} \times X$.

\end{itemize}

\bigskip

\noindent
Write:

\smallskip 
\begin{itemize}[leftmargin=.5cm]
\setlength\itemsep{3pt}
\item
$\Omega^\bullet_{\mathrm{dR}}\big(\widehat{X}; \, \mathfrak{g}\big) \;\defneq\; \Omega^\bullet_{\mathrm{dR}}\big(\widehat{X}\big) \otimes \mathfrak{g}$ 
for the de Rham complex of smooth differential forms on $\widehat X$ with coefficients in $\mathfrak{g}$,
\item
$[-,-] \,:\, \Omega^\bullet_{\mathrm{dR}}\big(\widehat{X}; \mathfrak{g}) \otimes \Omega^\bullet_{\mathrm{dR}}\big(\widehat{X};\,  \mathfrak{g}\big)  \to \Omega^\bullet_{\mathrm{dR}}\big(\widehat{X};\,  \mathfrak{g}\big)$ for the induced super-Lie bracket, given by the wedge product of differential 
forms in the given order, tensored with the Lie bracket of their coefficients,
\item
$\langle -,-\rangle \,:\, \Omega^\bullet(X; \mathfrak{g}_P) \otimes \Omega^\bullet(X; \mathfrak{g}_P) \to \Omega^\bullet(X)$ for the induced 
graded pairing, given by the wedge product of differential forms in the given order tensored with the pairing of their coefficients -- and 
then regarded as a plain differential form on $X$, via the $\mathrm{ad}$-invariance of the pairing,
\item
$\Gamma_{T X} \big(TP/G\big) \,=\, \big\{ A \,\in\, \Omega^1_{\mathrm{dR}}(\widehat X) \otimes \mathfrak{g} \,\big\vert\, \forall_{i j}\; A_j 
\,=\, 
\differential g_{i j} + 
\mathrm{Ad}_{g_{i j}}(A_i) \;\mbox{on}\; U_i \cap U_j   \big\}$ for the set of gauge potentials in the sector $P$ (principal connections on $P$),
\item
$\Omega^\bullet_{\mathrm{dR}}(X; \mathfrak{g}_P) \;=\, \big\{\omega \,\in\,  \Omega^\bullet_{\mathrm{dR}}(\widehat{X}) \otimes \mathfrak{g} \,\big\vert\, \forall_{i,j} \;\; \omega_j = \mathrm{Ad}_{g_{i j}}(\omega_i) \;\mbox{on}\; U_{i} \cap U_j \big\}$ for the de Rham complex of smooth 
differential forms on $X$ with values in sections of the $P$-adjoint bundle,
\item
$\differential_A \,:\,\Omega^\bullet_{\mathrm{dR}}(X;\mathfrak{g}_P) \to \Omega^{\bullet }_{\mathrm{dR}}(X;\mathfrak{g}_P)$ for the
covariant de Rham differential (of degree=1) with respect to a given $A \,\in\, \Gamma_{T X}(T P / G)$: $\differential_A \omega \,\defneq\, \differential\, \omega + [A , \omega]$,

\item
$F_A \,\defneq\, \differential \, A + \tfrac{1}{2}[A, A] \,\in\, \Omega^2(X;\, \mathfrak{g}_P)$ for the magnetic flux density of $A$ (the curvature form).

\end{itemize}

\medskip

\noindent
Now: 

\newpage 
\noindent
The phase space of $G$-Yang-Mills theory in temporal gauge on $\mathbb{R}^{0,1}\times X$ with respect to the background field sector
$P$ is globally coordinatized by:

\begin{itemize}[leftmargin=.5cm]

\item 
$A \,\in\, \Gamma_{T X}(TP/G)$, the gauge potential in temporal gauge, serving as the canonical coordinate,

\item $E \,\in\, \Omega^2_{\mathrm{dR}}(X; \mathfrak{g}_P)$ the electric flux density, constituting the canonical momentum,
\end{itemize} 
hence with the non-trivial Poisson bracket being
\begin{equation}
  \label{CanonicalPoissonBracket}
  \Big\{
    \textstyle{\int_X} 
    \big\langle
      \omega
      ,
      E
    \big\rangle
    ,\,
    A(\hat x)
  \Big\}
  \;\;
  =
  \;\;
  \omega(\hat x)
  \hspace{1.2cm}
  \mbox{for}
  \;
    \omega 
      \in 
    \Omega^1_{\mathrm{dR}}(X; \mathfrak{g}_P)_{\mathrm{cpt}}
\end{equation}
\noindent
subject to a first-order constraint:
\begin{itemize}[leftmargin=.5cm]
\setlength\itemsep{2pt}
\item the {\it Gau{\ss}-Law}
\begin{equation}
  \label{GaussLaw}
  \differential_A E \;=\; 0
  \,.
\end{equation}
\end{itemize}

\medskip

\noindent
{\bf Linear flux observables in Yang-Mills theory.}
Consider in addition:

\begin{itemize}[leftmargin=.5cm]
\setlength\itemsep{2pt}
\item
$\Sigma \hookrightarrow X$ a closed oriented 2-dimensional submanifold (not necessarily connected), being the surface through which
electromagnetic flux is to be observed,
\item
$\Phi_E^\alpha \,:= \textstyle{\int_\Sigma} \langle \alpha, E \rangle$ for $\alpha \,\in\, \Omega^0_{\mathrm{dR}}(\Sigma, \mathfrak{g}_P)$, 
the electric flux through $\Sigma$ integrated against a weight $\alpha$,
\item
$\Phi_B^\beta := \textstyle{\int_\Sigma} \langle \beta, F_A \rangle$ for $\beta \in \Omega^0_{\mathrm{dR}}(\Sigma, \mathfrak{g}_P)$, 
the magnetic flux through $\Sigma$ integrated against weight $\beta$,
\end{itemize}
where the weights (or: ``smearing functions'') are smooth Lie-algebra valued functions, precisely:
\begin{equation}
  \label{SmearingFunction}
  \alpha,
  \, 
  \beta
  \;\in\;
  \Omega^0_{\mathrm{dR}}(\Sigma; \mathfrak{g}_P)
  \;\simeq\;
  \Gamma_\Sigma\big(
    \mathfrak{g} \rtimes_{\mathrm{ad}} P
  \big)
  \,.
\end{equation}

\smallskip 
The subtlety pointed out and resolved in \cite{CattaneoPerez17} is that these $\Phi^\alpha_E$, $\Phi^\beta_{B}$ are not technically observables on the 
phase space, since their would-be associated Hamiltonian vector fields are not smooth; but that gauge-equivalent regularized 
observables are obtained by considering:
\smallskip 
\begin{itemize}[leftmargin=.5cm]
\setlength\itemsep{2pt}
\item
$\widehat{\Sigma} \hookrightarrow X$ the exterior component of a tubular neighborhood of $\Sigma$ in $X$, hence a non-compact 
3-dimensional submanifold with boundary $\partial \widehat{\Sigma} = \Sigma$;
\item
$\widehat{\Phi}_E^\alpha \,:=\, \textstyle{\int_{\widehat{\Sigma}}} \langle \differential_A \alpha, E \rangle$ for 
compactly supported $\alpha \,\in\, \Omega^0_{\mathrm{dR}}\big(\widehat{\Sigma};\mathfrak{g}_P\big)_{\mathrm{cpt}}$, the electric flux observable;
\item
$\widehat{\Phi}_E^\alpha \,:=\, \textstyle{\int_{\widehat{\Sigma}}} \langle \differential_A \beta, F_A \rangle$ for compacty supported $\beta \,\in\, \Omega^0_{\mathrm{dR}}\big(\widehat{\Sigma};\mathfrak{g}_P\big)_{\mathrm{cpt}}$, the magnetic flux observable.
\end{itemize}

\smallskip

\noindent
The above two forms of the magnetic flux observable are actually equal, due to the Bianchi identity,
\begin{equation}
  \def\arraystretch{1.5}
  \begin{array}{ll}
    \mathllap{\widehat{\Phi}_B^\beta}
    \;\defneq\;
    \textstyle{\int_{\widehat{\Sigma}}}
    \,
    \langle 
    \differential_A \beta , F_A \rangle
    &
    \proofstep{
      by definition
    }
    \\
    \;=\;
    \textstyle{\int_{\widehat{\Sigma}}}
    \,
    \langle 
    \differential_A \beta , F_A \rangle
    +
    \textstyle{\int_{\widehat{\Sigma}}}
    \,
    \langle 
     \beta
     ,\, 
     \underbrace{
       \differential_A F_A 
     }_{ 0 }
     \rangle
     &
     \proofstep{
       Bianchi ident.
     }
     \\[-10pt]
     \;=\;
     \textstyle{\int_{\widehat{\Sigma}}}
     \,
     \differential
     \,
     \langle
       \beta,\, F_A
     \rangle
     &
     \proofstep{
       Leibniz rule \& ad-invariance
     }
     \\
     \;=\;
     \textstyle{\int_{\Sigma}}
     \,
     \differential
     \langle
       \beta,\, F_A
     \rangle
     &
     \proofstep{
       by invariance
     }
     \\
     \;=\;
     \Phi_B^\beta
     &
     \proofstep{
       by definition
       ,
     }
  \end{array}
\end{equation}
and the analogous computation, but now using the Gau{\ss} law \eqref{GaussLaw} in place of the Bianchi identity, shows that the two forms of the electric flux observables coincide on the
constraint surface (i.e. up to a term proportional to $\differential_A E$):
\begin{equation}
  \widehat{\Phi}_E^\alpha
  \;\approx\;
  \Phi_E^\alpha
  \,.
\end{equation}

With the canonical Poisson bracket \eqref{CanonicalPoissonBracket}, one finds the Poisson brackets of these regularized
electromagnetic linear flux observables, first for electric/electric fluxes (cf. \cite[(7)]{CattaneoPerez17})
\begin{equation}
  \def\arraystretch{1.8}
  \begin{array}{ll}
    \Big\{
      \widehat{\Phi}_E^\alpha
      ,\,
      \widehat{\Phi}_E^\beta
    \Big\}
    \;\defneq\;
    \Big\{
    \textstyle{\int_{\widehat{\Sigma}}}
    \,
    \langle
      \differential_A
      \alpha
      ,\,
      E
    \rangle
    ,\,
    \textstyle{\int_{\widehat{\Sigma}}}
    \,
    \langle
      \differential_A
      \beta
      ,\,
      E
    \rangle
    \Big\}
    &
    \proofstep{
      by definition
    }
    \\
    \;=\;
    \Big\{
    \textstyle{\int_{\widehat{\Sigma}}}
    \,
    \langle
      \differential_A
      \alpha
      ,\,
      E
    \rangle
    ,\,
    \textstyle{\int_{\widehat{\Sigma}}}
    \,
    \big\langle
      [A,
      \beta]
      ,\,
      E
    \big\rangle
    \Big\}
    +
    \Big\{
    \textstyle{\int_{\widehat{\Sigma}}}
    \,
    \big\langle
      [A,\alpha]
      ,\,
      E
    \big\rangle
    ,\,
    \textstyle{\int_{\widehat{\Sigma}}}
    \,
    \langle
      \differential_A
      \beta
      ,\,
      E
    \rangle
    \Big\}
    &
    \proofstep{
      expanding
    }
    \\
    \;=\;
    \textstyle{\int_{\widehat{\Sigma}}}
    \,
    \Big\langle
      [
        \differential_A \alpha
        ,\,
        \beta
      ]
      ,\,
      E
    \Big\rangle
    -
    \textstyle{\int_{\widehat{\Sigma}}}
    \,
    \Big\langle
      [
        \alpha
        ,\,
        \differential_A \beta
      ]
      ,\,
      E
    \Big\rangle
    &
    \proofstep{
      by \eqref{CanonicalPoissonBracket}
    }
    \\
    \;=\;
    \textstyle{\int_{\widehat{\Sigma}}}
    \Big\langle
      \differential_A
      [\alpha,\, \beta]
      ,\,
      E
    \Big\rangle
    \;\defneq\;
    \widehat{\Phi}_E^{[\alpha, \beta]}
    &
    \proofstep{
      Leibniz rule,
    }
  \end{array}
\end{equation}
and for the electric/magnetic fluxes (a statement which seems not to have been recorded before):
\begin{equation}
  \def\arraystretch{1.8}
  \begin{array}{ll}
    \Big\{
      \widehat{\Phi}_E^\alpha
      ,\,
      \widehat{\Phi}_B^\alpha
    \Big\}
    \;\defneq\;
    \Big\{
    \textstyle{\int_{\widehat{\Sigma}}}
    \,
    \langle
      \differential_A
      \alpha
      ,\,
      E
    \rangle
    ,\,
    \textstyle{\int_{\widehat{\Sigma}}}
    \,
    \langle
      \differential_A
      \beta
      ,\,
      F_A
    \rangle
    \Big\}
    &
    \proofstep{
      by definition
    }
    \\
    \;\defneq\;
    \Big\{
    \textstyle{\int_{\widehat{\Sigma}}}
    \,
    \langle
      \differential_A
      \alpha
      ,\,
      E
    \rangle
    ,\,
    \textstyle{\int_{\widehat{\Sigma}}}
    \,
    \big\langle
      \differential\, \beta
      +
      [A,
      \beta]
      ,\,
      \differential A
      + 
      \tfrac{1}{2}[A,A]
    \big\rangle
    \Big\}
    &
    \proofstep{
      expanding
    }
    \\
    \;=\;
    \textstyle{\int_{\widehat{\Sigma}}}
    \,
    \big\langle
      [\differential_A \alpha, \beta]
      ,\,
      F_A
    \big\rangle
    +
    \textstyle{\int_{\widehat{\Sigma}}}
    \,
    \big\langle
      \differential_A \beta
      ,\,
      \differential_A
      \, 
      \differential_A 
      \,
      \alpha
    \big\rangle
    &
    \proofstep{
      by \eqref{CanonicalPoissonBracket}
    }
    \\
    \;=\;
    \textstyle{\int_{\widehat{\Sigma}}}
    \,
    \big\langle
      [\differential_A \alpha, \beta]
      ,\,
      F_A
    \big\rangle
    +
    \textstyle{\int_{\widehat{\Sigma}}}
    \,
    \big\langle
      \differential_A \beta
      ,\,
      [F_A
      ,\,
      \alpha]
    \big\rangle
    &
    \proofstep{
      by definition
    }
    \\
    \;=\;
    \textstyle{\int_{\widehat{\Sigma}}}
    \,
    \big\langle
      [\differential_A \alpha, \beta]
      ,\,
      F_A
    \big\rangle
    +
    \textstyle{\int_{\widehat{\Sigma}}}
    \,
    \big\langle
      [\alpha,\,\differential_A \beta]
      ,\,
      F_A
    \big\rangle
    &
    \proofstep{
      cycl. invariance
    }
    \\
    \;=\;
    \textstyle{\int_{\widehat{\Sigma}}}
    \,
    \big\langle
      \differential_A 
      [\alpha, \beta]
      ,\,
      F_A
    \big\rangle
    \;\defneq\;
    \widehat{\Phi}_B^{[\alpha, \beta]}
    &
    \proofstep{
      Leibniz rule
      ,
    }
  \end{array}
\end{equation}
where we used that 
$
  \langle
    -
    ,
    [-
    ,
    -]
  \rangle
  \,:\,
  \mathfrak{g}
  \otimes
  \mathfrak{g}
  \otimes
  \mathfrak{g}
  \to 
  \mathbb{R}
$
is invariant under cyclic permutations.

\medskip

This completes the proof of Thm. \ref{PhaseSpaceOfYangMillsFluxes}.

\subsection{Homotopy theory of topological field sectors}
\label{HomotopyTheoryOfTopologicalFieldSectors}

\noindent
{\bf Topology of fields vanishing at infinity.}
As usual in algebraic topology, we work in the category $\mathrm{kTopSp}$ of compactly-generated topological spaces (for pointers 
see \cite[p. 21]{SS21EquivariantBundles}, and we will just say ``topological spaces'', for short), where for $X, Y, Z \,\in\, \mathrm{kTopSp}$ 
the mapping spaces $\mathrm{Map}(-,-)$ and the product spaces $(-) \times (-)$ are related (``Cartesian closure'') by natural 
homeomorphisms of the form
\medskip 
$$
  \mathrm{Map}
  \big(
    X \times Y
    ,\,
    Z
  \big)
  \;\simeq\;
  \mathrm{Maps}
  \big(
    X
    ,\,
    \mathrm{Maps}(Y,\, Z)
  \big)
  \,.
$$
This property is inherited by the category of pointed spaces 
$$
 \mathrm{kTopSp}^{\ast/} 
  \;:=\; 
  \big\{
    X \in \mathrm{kTop}
    ,\, 
    \infty_X \in X
  \big\}
$$
(which here we think of as spaces equipped with a ``point at infinity'', see around \cite[Ntn. 3.3]{SS23Mf} for more), now with 
respect to the mapping sub-space $\mathrm{Map}^{\ast/}(-,-)$ of point-preserving maps  and the ``smash product'' $(-) \wedge (-)$, 
which identifies everything ``at infinity'' with a single point at infinity:
\bigskip 
\begin{equation}
  \label{SmashProduct}
  X
  ,
  Y
  \;\in\;
  \mathrm{kTopSp}^{\ast}
  \hspace{1cm}
  \vdash
  \hspace{1cm}
  X \wedge Y
  \;:=\;
  \frac{
    X \times Y
  }{
    X \!\times\! \{\infty_X\}
    \;\cup\;
    \{\infty_Y\} \!\times\! Y
  }
  \,,
\end{equation}
in that
\medskip 
\begin{equation}
  \label{InternalHomIsoForPointedSpaces}
  \mathrm{Map}^{\ast/}\big(
    X \wedge Y
    ,\,
    Z
  \big)
  \;\simeq\;
  \mathrm{Map}^{\ast/}\big(
    X
    ,\,
    \mathrm{Map}^{\ast/}(
      Y
      ,\,
      Z
    )
  \big).
\end{equation}

\medskip

There are several ways to turn a topological space $X$ into a pointed topological space. We write:
\begin{itemize}
\item
  $X_{\plus}$ for $X$ with a disjoint base point adjoined, so that none of the original points of $X$ is ``at infinity''. 
  Accordingly, preserving a disjoint point at infinity is no extra condition on maps:
  \medskip 
\begin{equation}
  \label{DisjointBasePointIsLeftAdjoint}
  \mathrm{Map}^{\ast/}\big(
    X_\plus
    ,\,
    -
  \big)
  \;\simeq\;
  \mathrm{Map}(
    X
    ,\,
    -
  )
  \,.
\end{equation}

\item
  $X_\compact$ for $X$ with a basepoint adjoined whose open neighborhoods are the complements of closed compact subsets of $X$ 
  (called the ``Alexandroff one-point compactification'' of $X$, cf. \cite[pp. 5]{Cutler20}). This means that continuous paths may 
  reach the ``point at infinity''.

\item $X$ for connected spaces, regarded as pointed by any one of their points (as on the right of
\eqref{OnePointCompactificationOfEuclideanSpace} and \eqref{BasedLoopSpace} below). 
\end{itemize}

\smallskip

\noindent
For example: 
\begin{itemize}[leftmargin=.5cm]
\item
  Identifying the ``ends'' (cf. \cite{Peschke90}) of Euclidean space with a point at infinity yields a sphere (cf. \cite[p. 7]{SS23Mf}):
  \medskip 
  \begin{equation}
    \label{OnePointCompactificationOfEuclideanSpace}
    (\mathbb{R}^{n \geq 1})_{\compact}
    \underset{
      \mathrm{homeo}
    }{\;\;\;\simeq\;\;\;}
    S^n
    .
  \end{equation}
\item
Adjoining the point at infinity to a product is the smash product of the factor with their separate points-at-infinity \cite[Prop. 1.6]{Cutler20}:
\begin{equation}
  \label{OnePointCompactificationTakesProductsToSmashProducts}
  (
    X \times Y
  )_{\compact}
  \;\simeq\;
  X_\compact \wedge Y_\compact
  \,.
\end{equation}
For example, with \eqref{OnePointCompactificationOfEuclideanSpace} this gives
\begin{equation}
  \label{SmashProductOfSpheres}
  S^{n_1}
  \wedge
  S^{n_2}
  \;\;
  \simeq
  \;\;
  S^{n_1 + n_2}.
\end{equation}
\item If a space $\Sigma$ is already compact, then the adjoined point at infinity is disjoint:
\begin{equation}
  \label{AdjoiningPointAtInfinityToCompactSpace}
  \mbox{$\Sigma$ compact}
  \hspace{1cm}
   \Rightarrow
  \hspace{1cm}
  \Sigma_{\compact}
  \;\simeq\;
  \Sigma_{\plus}\;.
\end{equation}

\item
  The (reduced) suspension
  
  \begin{equation}
    \label{Suspension}
    S^1 \wedge X_{\plus}
    \;\defneq\;
    \frac{
      \mathbb{R}^1_{\compact}
      \times
      X
    }{
      \{\infty\} \times X 
    }
  \end{equation}
  is to be thought of as the cylinder $\mathbb{R}^1 \times X$ with both ends regarded as being at infinity.
\item
If $\Sigma$ is already compact, then its suspension \eqref{Suspension}
is equivalently the compactification of its product with the real line:
\begin{equation}
  \label{CompactificationOfProductOfCompactSpaceWithLine}
  \def\arraystretch{1.2}
  \begin{array}{lll}
    \big(
      \mathbb{R}^1 
      \times
      \Sigma
    \big)_{\compact}
    &\simeq\;
      \mathbb{R}^1_\compact 
      \wedge
      \Sigma_\compact
    &
    \proofstep{by
      \eqref{OnePointCompactificationTakesProductsToSmashProducts}
    }
    \\
    &\simeq\;
      S^1
      \wedge
      \Sigma_\plus
    & 
    \proofstep{by
      \eqref{OnePointCompactificationOfEuclideanSpace}
      \&
      \eqref{AdjoiningPointAtInfinityToCompactSpace}.
    }
  \end{array}
\end{equation}
\item
The based loop space of any $X \,\in\, \mathrm{kTopSp}^{\ast/}$ is the pointed mapping space from the circle:
\begin{equation}
  \label{BasedLoopSpace}
  \Omega X
  \;=\;
  \mathrm{Map}^{\ast/}(
    S^1
    ,\,
    X
  )
  \,,
\end{equation}
where the circle is equipped with any basepoint.

\item
The $n$-th homotopy group of $X \,\in\, \mathrm{kTopSp}^{\ast/}$ is the connected components 
of the pointed mapping space out of the $n$-sphere:
\medskip 
\begin{equation}
  \label{HomotopyGroupsViaPointedMappingSpaces}
  \pi_n(X, \infty_X)
  \;=\;
  \pi_0
  \,
  \mathrm{Map}^{\ast/}\big(
    S^n
    ,\,
    X
  \big)
  \,.
\end{equation}

\item With $G$ a topological group, its classifying space $B G$ (pointers in \cite[\S 2.3]{SS21EquivariantBundles}) 
is connected and its based loop space \eqref{BasedLoopSpace} 
is weakly homotopy equivalent \eqref{WeakHomotopyEquivalence} to the underlying space of $G$:
\begin{equation}
  \label{LoopSpaceOfClassifyingSpace}
  \Omega B G 
  \underset{\mathrm{whe}}{\;\simeq\;} 
  G
  \,.
\end{equation}
For example, 
\medskip 
\begin{equation}
  \label{CircleAsClassifyingSpaceOfIntegers}
  S^1 
    \underset{\mathrm{whe}}{\;\simeq\;} 
  B \mathbb{Z}\;.
\end{equation}
\end{itemize}

For classifying spaces of (topological sectors of) physical fields, we are to think of their point at infinity 
as classifying the {\it vanishing field}, because the constant pointed classifying map, which factors as 
\medskip 
$$
  X \xrightarrow{\quad} 
  \{\infty_{B G}\}
  \longhookrightarrow
  B G
  \,,
$$
classifies the trivial field. It is in this way that general pointed classifying maps literally {\it vanish at infinity} (cf. \cite[\S 2.1]{SS22Configurations}\cite[Rem. 2.3]{SS23Mf}):
\medskip 
\begin{equation}
  \label{ClasifyingMapsVanishingAtInfinity}
  \begin{tikzcd}[row sep=small]
    X 
    \ar[rr]
      && 
    B G
    \\
    \{\infty_X\}
    \ar[u, hook]
    \ar[rr]
    &&
    \{\infty_{B G}\}
    \mathrlap{\,.}
    \ar[u, hook]
  \end{tikzcd}
\end{equation}

For example, the familiar classification of Yang-Mills instanton sectors --- as $\mathrm{SU}(2)$-valued gauge fields on $\mathbb{R}^4$ 
which vanish at infinity --- is obtained as follows:
\begin{equation}
  \label{YMInstantonSectors}
  \def\arraystretch{1.3}
  \begin{array}{lll}
    \pi_0
    \,
    \mathrm{Map}^{\ast/}\big(
      \mathbb{R}^4_{\compact}
      ,\,
      B \mathrm{SU}(2)
    \big)
    &
   \simeq\;
    \pi_0
    \,
    \mathrm{Map}^{\ast/}\big(
      S^4
      ,\,
      B \mathrm{SU}(2)
    \big)
    &
    \proofstep{by 
      \eqref{OnePointCompactificationOfEuclideanSpace}
    }
    \\
    &\simeq\;
    \pi_0
    \,
    \mathrm{Map}^{\ast/}\big(
      S^3 \wedge S^1
      ,\,
      B \mathrm{SU}(2)
    \big)
    &
    \proofstep{by 
      \eqref{OnePointCompactificationTakesProductsToSmashProducts}
    }
    \\
    &\simeq\;
    \pi_0
    \,
    \mathrm{Map}^{\ast/}\Big(
      S^3 
      ,\,
      \mathrm{Map}^{\ast/}\big(
        S^1
        ,\,
        B \mathrm{SU}(2)
      \big)
   \! \Big)
    &
    \proofstep{by 
      \eqref{InternalHomIsoForPointedSpaces}
    }
    \\
    &\simeq\;
    \pi_0
    \,
    \mathrm{Map}^{\ast/}\big(
      S^3 
      ,\,
      \underbrace{\mathrm{SU}(2)}_{
        \simeq S^3
      }
    \big)
    &
    \proofstep{by 
      \eqref{LoopSpaceOfClassifyingSpace}
    }
    \\[-3pt]
    &\simeq\;
    \pi_3(S^3)
    \;=\;
    \mathbb{Z}
    &
    \proofstep{by \eqref{HomotopyGroupsViaPointedMappingSpaces}.
    }
  \end{array}
\end{equation}

\medskip

\noindent
{\bf Homotopy theory of topological field sectors.}
A {\it weak homotopy equivalence}  is a continuous map that induces isomorphisms on 
all homotopy groups \eqref{HomotopyGroupsViaPointedMappingSpaces}:
\begin{equation}
  \label{WeakHomotopyEquivalence}
  f \,:\,
  X 
    \xrightarrow[\; \mathrm{whe} \;]{\sim}
  Y
  \hspace{.8cm}
  \Leftrightarrow
  \hspace{.8cm}
  \left\{\!\!\!
  \def\arraystretch{1.5}
  \begin{array}{rcl}
  \pi_0(f)
  \,:\,\pi_0(X)
  &
  \xrightarrow{\sim}
  &
  \pi_0(Y),
  \\
  \underset{
    \scalebox{.8}{$
    {
    n \in \mathbb{N}_{\geq 1}
    }
    \atop
    { x \in X } 
    $}
  }{\forall}
  \;
  \pi_n(f,x)
  \,:\,
  \pi_n(X,x)
  &
  \xrightarrow{\sim}
  &
  \pi_n\big( Y, f(y) \big)
  \mathrlap{\,,}
  \end{array}
  \right.
\end{equation}

\vspace{1mm} 
\noindent making $X$ and $Y$ be ``the same for all purposes'' of topological homotopy theory (see \cite[Ex. 1.1]{Char} for review).

\smallskip

\noindent
For example: 

\begin{itemize}[leftmargin=.5cm]

\item 
The smash product \eqref{SmashProduct}
of smooth manifolds with contractible smooth manifolds\footnote{
  The assumption in \eqref{SmashProductWithContractibleManifold} that $X$ and $Y$ be smooth manifolds is not necessary for this statement,
  we make it only for brevity of the discussion.
  A sufficient condition is that $X$ and $Y$ admit the structure of CW-complexes (which is the case for smooth manifolds by the triangulation theorem).
} is the identity up to weak homotopy equivalence, and the pointed mapping space construction is insensitive, up to weak homotopy equivalence,
to maps out of contractible manifolds:
\begin{equation}
  \label{SmashProductWithContractibleManifold}
  \left.
  \def\arraystretch{1.4}
  \begin{array}{l}
  X, Y \,\in\, \mathrm{SmthMfd}
  ,
  \\
  X 
    \underset{\mathrm{whe}}{\;\simeq\;} 
  \ast
  \end{array}
  \right\}
  \hspace{.8cm}
    \Rightarrow
  \hspace{.8cm}
  \left\{
  \def\arraystretch{1.8}
  \begin{array}{l}
    X \wedge Y
      \underset{\mathrm{whe}}{\;\simeq\;} 
    Y 
    \,,
    \\
    \mathrm{Map}^{\ast/}(
      X \wedge Y
      ,\,
      -
    )
      \underset{\mathrm{whe}}{\;\simeq\;} 
    \mathrm{Map}^{\ast/}(
      Y
      ,\,
      -
    )
    \mathrlap{\,.}
  \end{array}
  \right.
\end{equation}

\smallskip 
\item
If $\Sigma$ and $Y$ have the structure of smooth manifolds with $\Sigma$ compact, then there is a weak homotopy equivalence
\eqref{WeakHomotopyEquivalence} from the Fr{\'e}chet manifold of smooth functions $\Sigma \xrightarrow{\scalebox{0.6}{$ \mathrm{smooth}$}} Y$ 
to the mapping space of the underlying topological spaces --- an instance of the 
{\it smooth Oka principle} \cite[Thm. 3.3.63]{SS21EquivariantBundles}: 
\medskip 
\begin{equation}
  \label{SmoothOkaPrinciple}
  C^\infty\big(
    \Sigma
    ,\,
    Y
  \big)
  \underset{
    \mathrm{whe}
  }{\;\simeq\;}
  \mathrm{Map}(\Sigma,\, Y)
  \,.
\end{equation}
As such, this depends only on the weak homotopy type of $\Sigma$ and $Y$ themselves. For instance,
if $Y \,\defneq\, \mathbb{R}^n/\mathbb{Z}^n$ is a torus, whose weak homotopy type is that of the 
classifying space $B \mathbb{Z}^n$ \eqref{CircleAsClassifyingSpaceOfIntegers}, then 
\medskip 
\begin{equation}
  \label{SmoothOkaIntoTorus}
  C^\infty\big(
    \Sigma
    ,\,
    \mathbb{R}^n/\mathbb{Z}^n
  \big)
  \underset{
    \mathrm{whe}
  }{\;\simeq\;}
  \mathrm{Map}(
    \Sigma,
    \, 
    B \mathbb{Z}^n
  )
  \,.
\end{equation}
\end{itemize}

\medskip

\noindent
{\bf Cohomology classifying topological field sectors.}
That the space $B G$ \eqref{LoopSpaceOfClassifyingSpace} is ``classifying'' refers to the homotopy classes of maps into it corresponding to isomorphism 
classes of $G$-principal bundles over smooth manifolds (at least), hence to the degree=1 non-abelian cohomology with coefficients in $G$ (cf. \cite[Ex. 2.2]{Char}):
\smallskip 
\begin{equation}
  \pi_0
  \,
  \mathrm{Map}\big(
    \Sigma
    ,\,
    B G
  \big)
  \;\simeq\;
  H^1(\Sigma;\, G)\;.
\end{equation}

In the case when $G \,\defneq\, A$ is abelian there is (abelian) topological group structure on $B A$ itself so that we iteratively 
obtain higher classifying space $B^{n+1} A \,:=\, B^n A$. For $A$ discrete, these are ``Eilenberg-MacLane spaces'' $K(A,n)$ which 
classify ordinary cohomology in higher degrees (cf. \cite[Ex. 2.1]{Char}):
\medskip 
\begin{equation}
  \label{OrdinaryCohomology}
  \pi_0
  \,
  \mathrm{Map}\big(
    \Sigma
    ,\,
    B^n A
  \big)
  \;\simeq\;
  H^n(\Sigma;\, A)
  \,.
\end{equation}

Notice that, under this equivalence, the usual group structure on ordinary cohomology comes from the pointwise group structure of maps into a topological group.

Hence in general, for any topological group $G$, we may think of maps into a $B G$ as 1-cocycles of (possibly non-abelian) $G$-cohomology, 
and of homotopies between such maps as coboundaries between the corresponding cocycles. In this sense, the mapping space into $B G$ is 
the {\it cocycle space} of $G$-cohomology (\cite[Def. 2.1]{Char}) and its connected components are the (non-abelian) cohomology classes:

\begin{equation}
  \label{NonabelianCohomology}
  \mathrm{Map}(
    \Sigma
    ,\,
    B G
  )
  \;\;\;
    =
  \;\;\;
  \left\{
  \begin{tikzcd}
    \Sigma
    \ar[
      rr,
      bend left=50,
      "{ c }"{description, name=s},
      "{
        \scalebox{.7}{
          \color{gray}
          cocycle
        }
      }"
    ]
    \ar[
      rr,
      bend right=50,
      "{ c' }"{description, name=t},
      "{
        \scalebox{.7}{
          \color{gray}
          cocycle
        }
      }"{swap}
    ]
    \ar[
      from=s,
      to=t,
      shorten=2pt,
      Rightarrow,
      "{
        \scalebox{.7}{
          \color{gray}
          coboundary
        }
      }"{description}
    ]
    &&
    B G
  \end{tikzcd}
  \right\}
  ,\,\hspace{.5cm}
  H^1(\Sigma;\, G)
  \;\;
    :=
  \;\;
  \pi_0\,
  \mathrm{Map}(
    \Sigma
    ,\,
    B G
  )
  \,.
\end{equation}

\smallskip 
Similarly, the loop space of the mapping space into $B G$ is the mapping space into $G$ 
\begin{equation}
  \label{LoopSpaceOfMappingIntoBG}
  \def\arraystretch{1.8}
  \begin{array}{lll}
    \Omega
    \,
    \mathrm{Map}\big(
      \Sigma
      ;\,
      B G
    \big)
    &
    \simeq\;
    \mathrm{Map}^{\ast/}\Big(
      S^1
      ,\,
    \mathrm{Map}^{\ast/}\big(
      \Sigma_\plus
      ;\,
      B G
    \big)
    \Big)
    &
    \proofstep{
     by
     \eqref{BasedLoopSpace}
     \&
     \eqref{DisjointBasePointIsLeftAdjoint}
    }
    \\
    &\;\simeq\;
    \mathrm{Map}^{\ast/}\Big(
      \Sigma_+
      ,\,
    \mathrm{Map}^{\ast/}\big(
      S^1
      ;\,
      B G
    \big)
    \Big)
    &
    \proofstep{
      by
      \eqref{InternalHomIsoForPointedSpaces}
    }
    \\
    &\;\simeq\;
    \mathrm{Map}^{\ast/}\big(
      \Sigma_\plus
      ,\,
      G
    \big)
    &
    \proofstep{
      by
      \eqref{LoopSpaceOfClassifyingSpace}
    }
    \\
    &\;\simeq\;
    \mathrm{Map}\big(
      \Sigma
      ,\,
      G
    \big)
    &
    \proofstep{
      by
      \eqref{DisjointBasePointIsLeftAdjoint},
    }
  \end{array}
\end{equation}
whose connected components are the 0-cohomology with coefficients in $G$:
\begin{equation}
  \label{ZeroCohomologyWithCoefficientsInATopologicalGroup}
  H^0(\Sigma;\, G)
  \;\;
  :=
  \;\;
  \pi_0
  \,
  \mathrm{Map}(\Sigma,\, G)
  \,.
\end{equation}

\begin{remark}[Ordinary cohomology with topological group coefficients]
\label{OrdinaryCohomologyWithTopologicalGroupCoefficients}
Notice, with \eqref{CircleAsClassifyingSpaceOfIntegers}, that 
\smallskip 
\begin{equation}
  \label{ZerothCircleCohomology}
  H^0(
    \Sigma
    ;\,
    S^1
  )
  \;\simeq\;
  H^1(
    \Sigma
    ;\,
    \mathbb{Z}
  )
\end{equation}

\vspace{1mm} 
\noindent and beware that the usual notation ``$H^n(\Sigma; \mathrm{U}(1))$'' tacitly refers to the circle 
coefficient understood with its 
{\it discrete} topology, hence is quite different. To make this explicit, if we write $\flat A$ for the underlying discrete group of 
a topological abelian group, then the usual notion of ordinary cohomology with coefficients in $A$ is $H^n(\Sigma; \flat A)$ 
in the above notation \eqref{OrdinaryCohomology}.
\end{remark}

\medskip
The following derivations are standard for homotopy theorists but may serve as instructive examples of the above notions for other readers :
\begin{lemma}
\begin{equation}
  \label{MappingSpaceIntoCircle}
  \mathrm{Map}(\Sigma, S^1)
  \underset{
    \mathrm{whe}
  }{\;\simeq\;}
  H^1(
    \Sigma
    ;\,
    \mathbb{Z}
  )
  \times
  B
  \big(
  H^0(
    \Sigma
    ;\,
    \mathbb{Z}
  )
  \big)
  .
\end{equation}
\end{lemma}
\begin{proof}
First, observe that the connected components are
\[
  \def\arraystretch{1.3}
  \begin{array}{lll}
    \pi_0
    \,
    \mathrm{Map}(
      \Sigma
      ,\,
      S^1
    )
    &\simeq\;
    H^0\big(
      \Sigma
      ;\,
      S^1
    \big)
    &
    \proofstep{
      by
      \eqref{OrdinaryCohomology}
    }
    \\[-1pt]
    &
    \simeq\;
    H^1\big(
      \Sigma
      ;\,
      \mathbb{Z}
    \big)
    &
    \proofstep{
      by
      \eqref{ZerothCircleCohomology}.
    }
  \end{array}
\]
Moreover, the fundamental group at the basepoint is
\begin{equation}
  \label{FundamentalGroupOfMappingSpaceIntoCircle}
  \def\arraystretch{1.5}
  \begin{array}{lll}
    \pi_1\big(
      \mathrm{Map}(
        \Sigma
        ,\,
        S^1
      )
      ,\,
      0
    \big)
    &\simeq\;
    \pi_0\big(
      \Omega
      \,
      \mathrm{Map}(
        \Sigma
        ,\,
        S^1
      )
    \big)
    &
    \proofstep{
      by
      \eqref{HomotopyGroupsViaPointedMappingSpaces}
    }
    \\
    &\simeq\;
    \pi_0\big(
      \mathrm{Map}(
        \Sigma
        ,\,
        \mathbb{Z}
      )
    \big)
    &
    \proofstep{
      by
      \eqref{LoopSpaceOfMappingIntoBG}
      \&
      \eqref{CircleAsClassifyingSpaceOfIntegers}
    }
    \\
    &\simeq\;
    H^0(\Sigma;\,\mathbb{Z})
    &
    \proofstep{
      by
      \eqref{OrdinaryCohomology}.
    }
  \end{array}
\end{equation}
Also, all higher homotopy groups at the basepoint vanish:
\begin{equation}
  \def\arraystretch{1.7}
  \begin{array}{lll}
    \pi_{n \geq 2}\big(
      \mathrm{Map}(
        \Sigma
        ,\,
        S^1
      )
      ,\,
      0
    \big)
    &\simeq\;
    \pi_0\big(
      \Omega^{n\geq 2}
      \,
      \mathrm{Map}(
        \Sigma
        ,\,
        S^1
      )
    \big)
    &
    \proofstep{
      by
      \eqref{HomotopyGroupsViaPointedMappingSpaces}
    }
    \\
    &\simeq\;
    \pi_0\Big(
      \mathrm{Map}\big(
        \Sigma
        ,\,
        \underbrace{
          \Omega^{n\geq 2}
          \,
        S^1}_{
          \mathclap{
            \underset
              {\mathrm{whe}}
              {\simeq} 
            \ast
          }
        }
      \big)
    \Big)
    &
    \proofstep{
      by
      \eqref{LoopSpaceOfMappingIntoBG}
      \&
      \eqref{SmashProductOfSpheres}
    }
    \\[-10pt]
    &\simeq
    \ast
    \,.
  \end{array}
\end{equation}
Hence the connected component of the neutral element of $\mathrm{Map}(\Sigma,\, S^1)$ is weakly homotopy equivalent to 
$B\big( H^1(\Sigma;\, \mathbb{Z}) \big)$.

But the topological group structure on $\mathrm{Map}(\Sigma, \, S^1)$ implies that the multiplication operation of any of its
elements constitutes a homeomorphism from the connected component of that element to that of the neutral element. Therefore, 
 the underlying spaces of all connected components are isomorphic. This implies \eqref{MappingSpaceIntoCircle}.
\end{proof}

\smallskip
\noindent
{\bf Homology Pontrjagin algebras.}
We consider homology with complex coefficients, throughout.\footnote{Much further interesting structure appears 
when considering homology with integer coefficients, but this is beyond the intended scope of the present note.}
For $\Gamma \,\in\, \Groups(\Sets)$ a discrete group, regarded as the loop space of its classifying space, 
$\Gamma \underset{\mathrm{whe}}{\,\simeq\,} \Omega B G$, see \eqref{LoopSpaceOfClassifyingSpace},
its homology is just the linear span of the underlying set
\begin{equation}
  \label{HomologyOfDiscreteGroup}
    H_\bullet(
      \Gamma
      ;\,
      \mathbb{C}
    )
    \;=\;    H_0(
      \Gamma
      ;\,
      \mathbb{C}
    )
    \;\simeq\;
    \mathbb{C}[\Gamma]
    ,\,
    \hspace{1cm}
  \begin{tikzcd}[row sep=-3pt, column sep=tiny]
    \Gamma 
    \ar[
      rr,
      hook
    ]
    &&
    \mathbb{C}[\Gamma]
    \\
    \gamma &\longmapsto&
    \mathcal{O}(\gamma)
  \end{tikzcd}
\end{equation}
so that 
\medskip 
$$
  \mathbb{C}[
    \Gamma
  ]
  \;=\;
  \bigg\{
  \textstyle{
  \underset{
    \gamma \in \Gamma
  }{\sum}
  }
  \,
  c_\gamma \cdot
  \mathcal{O}(\gamma)
  \;\Big\vert\;
  c_\gamma \in \mathbb{C}
  , 
  \mbox{
    $\mathrm{supp}(c_{(-)})$
    is finite
  }
 \!\! \bigg\}
  \,.
$$
Furthermore,  the Pontrjagin product \eqref{HomologyPontrjaginAlgebraonLoopSpaceOfFields} is the convolution product of the 
group algebra (\cite[\S III.13]{Weyl31}, cf. \cite[\S 3.4]{FultonHarris91}):
\smallskip 
\begin{equation}
  \label{DiscreteConvolutionProduct}
  \begin{tikzcd}[row sep=0pt]
    \mathbb{C}[
      \Gamma
    ]
    \otimes_{{}_\ComplexNumbers}
    \mathbb{C}[
      \Gamma
    ]
    \ar[
      rr,
      "{
        (-) \cdot (-)
      }"
    ]
    &&
    \mathbb{C}[
      \Gamma
    ]
    \\
    \mathcal{O}(\gamma)
    \otimes
    \mathcal{O}(\gamma')
    &\longmapsto&
    \mathcal{O}(\gamma \cdot \gamma')
      \end{tikzcd}
\end{equation}
\medskip 
\begin{equation*}
    \def\arraystretch{1.6}
  \begin{array}{l}
  \bigg(
  \textstyle{
    \underset{\gamma \in \Gamma}{\sum}
    \,
    c_\gamma
    \cdot
    \mathcal{O}(\gamma)
  }
 \! \bigg)
  \cdot
  \bigg(
  \textstyle{
    \underset{\gamma' \in \Gamma}{\sum}
    \,
    c'_{\gamma'}
    \cdot
    \mathcal{O}(\gamma')
  }
 \! \bigg)
  \;
    =
  \;
  \bigg(
  \textstyle{
    \underset{\gamma'' \in \Gamma}{\sum}
    \,
    \Big(
      \,
      \textstyle{
        \underset{
          \gamma \in \Gamma
        }{\sum}
      }
      c_\gamma
      \cdot
      c'_{
        \gamma^{-1} 
        \cdot
        \gamma''
    }
    \Big)
    \cdot
    \mathcal{O}(\gamma'')
  }
 \! \bigg)
  \,.
  \end{array}
     \end{equation*}

     \smallskip 
\noindent 
Now, since homology is invariant under weak homotopy equivalence \eqref{WeakHomotopyEquivalence}, 
i.e., 
\smallskip 
\begin{equation}
  \label{HomologyInvariantUnderWHE}
  X \underset{
    \mathrm{whe}
  }{\;\simeq\;}
  Y
  \hspace{.8cm}
  \Rightarrow
  \hspace{.8cm}
  H_\bullet(
    X;\,
    \mathbb{C}
  )
  \simeq
  H_\bullet(
    Y
    ;\,
    \mathbb{C}
  )
  \,,
\end{equation}
and by the K{\"u}nneth Theorem (e.g. \cite[Cor. 3B.7]{Hatcher02})
\smallskip 
\begin{equation}
  \label{KunnethTheorem}
  H_n(
    X \times Y
    ;\,
    \mathbb{C}
  )
  \;
    \simeq
  \;\;\;
  \underset{
    \mathclap{
      n_1 + n_2 = n
    }
  }{\bigoplus}
\;\;   H_{n_1}(
    X;\, \mathbb{C}
  )
  \otimes_{{}_\mathbb{C}}
  H_{n_2}(
    Y;\, \mathbb{C}
  )\,,
\end{equation}
this remains the case in degree=0 for products with connected CW-complexes (whose homology in degree 0 is $\mathbb{C}$), 
such as with classifying spaces of abelian groups, $\Gamma \times B A \underset{\mathrm{whe}}{\,\simeq\,} \Omega B( \Gamma \times B A)$:
\begin{equation}
  \label{ZeroHomologyOfGammaTimesConnected}
  H_0(
    \Gamma \times B A
    ;\,
    \mathbb{C}
  )
  \;\simeq\;
  H_0(\Gamma;\, \mathbb{C})
  \,.
\end{equation}

\medskip
\noindent Using all this, we establish Thm. \ref{PontrjaginRingOfMaxwellFluxObservables} as follows:
$$
  \def\arraystretch{1.8}
  \begin{array}{ll}
    H_0\bigg(
     \mathrm{Map}^{\ast/}
     \Big(
       \mathbb{R}^{0,1}_{\plus}
       \wedge
       (\mathbb{R}^1 \times \Sigma)_{\compact}
       ,\,
       B \big(
          \CircleGroup^2
          \rtimes 
          K
       \big)
     \Big)
     \,;\,
     \mathbb{C}
    \bigg)
    \\
    \;\simeq\;
    H_0\Big(
      \Omega
      \,
      \mathrm{Map}\big(
        \Sigma
        ,\,
        B\big(
          \CircleGroup^2
          \rtimes K
        \big)
      \big)
      \,;\,
      \mathbb{C}
    \Big)
    &
    \proofstep{
      by
      \eqref{KKFieldsModuliSpaceAsLoopSpace}
      \&
      \eqref{HomologyInvariantUnderWHE}
    }
    \\
    \;\simeq\;
    H_0\Big(
      \mathrm{Map}\big(
        \Sigma
        ,\,
        \CircleGroup^2
        \rtimes K
      \big)
      \,;\,
      \mathbb{C}
    \Big)
    &
    \proofstep{
      by
      \eqref{TopologicalKKFieldsExpressedOnSigma}
      \&
      \eqref{HomologyInvariantUnderWHE}
    }
    \\
    \;\simeq\;
    H_0\Big(
      H^1\big(
        \Sigma
        ;\,
        \mathbb{Z}
      \big)^2
      \times
      B\big(
        H^0(\Sigma;\, \mathbb{Z})
      \big)
      \rtimes
      H^0(\Sigma;\, K)
      \,;\,
      \mathbb{C}
    \Big)
    &
    \proofstep{
      by
      \eqref{MappingSpaceIntoCircle}
    }
    \\
    \;\simeq\;
    H_0\Big(
      H^1\big(
        \Sigma
        ;\,
        \mathbb{Z}
      \big)^2
      \rtimes
      H^0(\Sigma;\, K)
      \,;\,
      \mathbb{C}
    \Big)
    &
    \proofstep{
      by
      \eqref{ZeroHomologyOfGammaTimesConnected}
    }
    \\
    \;\simeq\;
    \mathbb{C}\big[
      H^1\big(
        \Sigma
        ;\,
        \mathbb{Z}
      \big)^2
      \rtimes
      H^0(\Sigma;\, K)
    \big]
    &
    \proofstep{
      by
      \eqref{HomologyOfDiscreteGroup}.
    }
  \end{array}
$$

\end{document}